%% file: main.tex
\documentclass[11pt]{article}

\usepackage[english]{babel}

\usepackage[margin=1in]{geometry}

\usepackage{graphicx}
\usepackage{amsfonts,amsmath,amssymb,amsthm} 
\usepackage{comment}
\usepackage{xcolor}
\definecolor{ForestGreen}{rgb}{0.1333,0.5451,0.1333}
\definecolor{DarkRed}{rgb}{0.80,0,0}
\definecolor{Red}{rgb}{1,0,0}
\usepackage[linktocpage=true,
pagebackref=true,colorlinks,
linkcolor=DarkRed,citecolor=ForestGreen,
bookmarks,bookmarksopen,bookmarksnumbered]
{hyperref}
\usepackage{cleveref}
\usepackage{algorithm}
\usepackage[noend]{algpseudocode}
\usepackage{booktabs}
\usepackage{multirow}
\usepackage{color-edits}
\addauthor[Mingwei]{mw}{red}

\theoremstyle{plain}
\newtheorem{theorem}{Theorem}[section]
\newtheorem{lemma}[theorem]{Lemma}

\newtheorem{corollary}[theorem]{Corollary}

\theoremstyle{definition}
\newtheorem{definition}[theorem]{Definition}

\newcommand{\R}{\mathbb{R}}

\newcommand{\calM}{\mathcal{M}}

\newcommand{\calP}{\mathcal{P}}
\newcommand{\bfv}{\mathbf{v}}
\newcommand{\bfc}{\boldsymbol{c}}

\newcommand{\bfr}{\mathbf{r}}
\newcommand{\bfx}{\mathbf{x}}

\newcommand{\MMS}{\mathsf{MMS}}
\newcommand{\Sym}{\mathsf{Sym}}
\newcommand{\SW}{\mathsf{SW}}

\newcommand{\argmax}{\mathop{\arg\max}}
\newcommand{\argmin}{\mathop{\arg\min}}
\newcommand{\PSg}{PS\textsuperscript{g}}
\newcommand{\PSc}{PS\textsuperscript{c}}

\title{When is Truthfully Allocating Chores no Harder than Goods?\footnote{The authors are ordered alphabetically.}}

\author{
    Bo Li\thanks{The Hong Kong Polytechnic University. Emails: \texttt{comp-bo.li@polyu.edu.hk}, \texttt{fangxiao.wang@connect.polyu.hk}}
    \and
    Biaoshuai Tao\thanks{Shanghai Jiao Tong University. Email: \texttt{bstao@sjtu.edu.cn}}
    \and
    Fangxiao Wang\footnotemark[2]
    \and
    Xiaowei Wu\thanks{University of Macau. Emails: \texttt{\{xiaoweiwu, yc17423\}@um.edu.mo}}
    \and
    Mingwei Yang\thanks{Stanford University. Email: \texttt{mwyang@stanford.edu}}
    \and
    Shengwei Zhou\footnotemark[4]
}
\date{}

\begin{document}
\maketitle
\date{}

\input{abstract}
\input{intro}
\input{prelim}
\input{indiv-2-agent}
\input{div}
\input{bobw}

\section*{Acknowledgments}
We would like to thank the anonymous reviewers for their detailed feedback and comments that helped significantly improve the exposition of this paper.
Bo Li is partly funded by the Hong Kong SAR Research Grants Council (No. PolyU 15224823) and the Department of Science and Technology of Guangdong Province (No. 2023A1515010592).
The research of Biaoshuai Tao is supported by the National Natural Science Foundation of China (No. 62472271) and the Key Laboratory of Interdisciplinary Research of Computation and Economics (Shanghai University of Finance and Economics), Ministry of Education.
Xiaowei Wu is funded by the Science and Technology Development Fund (FDCT), Macau SAR (file no. 0147/2024/RIA2, 0014/2022/AFJ, 0085/2022/A, and 001/2024/SKL).

\bibliographystyle{alpha}
\bibliography{references}

\clearpage
\appendix
\input{appendix/indiv-char}
\input{appendix/proof-div}

\end{document}

%% file: abstract.tex
\begin{abstract}
    We study the problem of fairly and efficiently allocating a set of items among strategic agents with additive valuations, where items are either all indivisible or all divisible.
    When items are \emph{goods}, numerous positive and negative results are known regarding the fairness and efficiency guarantees achievable by \emph{truthful} mechanisms, whereas our understanding of truthful mechanisms for \emph{chores} remains considerably more limited.
    In this paper, we discover various connections between truthful good and chore allocations, greatly enhancing our understanding of the latter via tools from the former.

    For indivisible chores with two agents, by leveraging the observation that a simple bundle-swapping operation transforms several properties for goods including truthfulness to the corresponding properties for chores, we characterize truthful mechanisms and derive tight guarantees of various fairness notions achieved by truthful mechanisms.
    Moreover, for homogeneous divisible chores, by generalizing the above transformation to an arbitrary number of agents, we characterize truthful mechanisms with two agents, show that every truthful mechanism with two agents admits an \emph{efficiency ratio} of $0$, and derive a large family of \emph{strictly truthful}, \emph{envy-free (EF)}, and \emph{proportional} mechanisms for an arbitrary number of agents.
    Finally, for indivisible chores with an arbitrary number of agents having \emph{bi-valued} cost functions, we give an \emph{ex-ante} truthful, ex-ante \emph{Pareto optimal}, ex-ante EF, and \emph{ex-post envy-free up to one item} mechanism, improving the best guarantees for bi-valued instances by prior works.
\end{abstract}

%% file: intro.tex
\section{Introduction}

The problem of fair division concerns how to divide a set of items among $n$ agents in a ``fair'' way.
The study of this problem can be traced back to 1948 when \cite{steihaus1948problem} first formalized fairness and introduced the notion of \emph{proportionality (PROP)}, i.e., each agent receives at least $1/n$ of his total value for the resources.
Besides, another prominent fairness notion is \emph{envy-freeness (EF)}~\cite{foley1967resource}, which asserts that exchanging bundles does not make any agent strictly happier.
When items are divisible and homogeneous, uniformly allocating each item to agents guarantees both PROP and EF.
During the past decades, there has also been increasing interest in the discrete setting where items are indivisible, in which EF or PROP allocations might not exist.\footnote{Consider the case with two agents and one item.}
This gives rise to several relaxations of the above fairness notions including \emph{envy-freeness up to one item (EF1)} \cite{DBLP:conf/sigecom/LiptonMMS04,DBLP:conf/bqgt/Budish10}, \emph{envy-freeness up to any item (EFX)} \cite{DBLP:conf/ecai/GourvesMT14,DBLP:journals/teco/CaragiannisKMPS19}, and \emph{maximin share fairness (MMS)} \cite{DBLP:conf/bqgt/Budish10}, whose computation and existence are then extensively studied by follow-up works.

When mechanisms are implemented in practice, agents are incentivized to misreport their preferences provided that doing so would lead to a more desirable outcome from their perspective, which could potentially cause severe fairness and welfare losses.
This inspires the study of \emph{truthful} mechanisms for the resource allocation problem by a long line of literature~\cite{DBLP:journals/sigecom/BezakovaD05,DBLP:conf/wine/MarkakisP11,DBLP:conf/ijcai/AmanatidisBM16}.
The mainstream of recent research can be divided into two parts: The first one assumes all items to be \emph{goods}, which are positively valued by agents, while the other one assumes all items to be \emph{chores}, which are negatively valued.
We separately introduce them as follows.

\paragraph{Truthful Allocation of Goods.}
When the goods are indivisible, \cite{DBLP:conf/sigecom/AmanatidisBCM17} characterize the entire family of truthful mechanisms for two agents with \emph{additive} valuations, which implies that truthfulness is incompatible with any non-trivial fairness notion.
Nevertheless, this impossibility result can be bypassed by studying \emph{binary} valuations~\cite{DBLP:conf/wine/0002PP020,DBLP:conf/aaai/BabaioffEF21,DBLP:conf/aaai/BarmanV22}, allowing monetary transfer~\cite{DBLP:journals/geb/GokoIKMSTYY24}, or resorting to randomized mechanisms~\cite{DBLP:journals/corr/abs-2407-13634}.

For (homogeneous) divisible goods, \cite{DBLP:conf/sigecom/ColeGG13} give a truthful mechanism that achieves a $1/e$-approximation of \emph{Nash welfare} with some goods unallocated.
More recently, \cite{freeman2023equivalence} characterize the set of truthful mechanisms for two agents (subject to some mild technical conditions) and provide a large family of \emph{strictly truthful}, PROP, and EF mechanisms.

\paragraph{Truthful Allocation of Chores.}
In stark contrast to the fruitful outcomes for truthful good allocation, our understanding of truthful chore allocation still remains considerably more limited.
For indivisible chores, \cite{DBLP:journals/mp/AzizLW24} study a randomized mechanism that is truthful \emph{in expectation} and achieve an expected approximation ratio of $O(\sqrt{\log n})$ for MMS.
Moreover, \cite{DBLP:journals/eor/SunC25} give a fair and efficient randomized mechanism for \emph{restricted additive} valuations\footnote{Restricted additive valuations constitute a strict superset of binary valuations and a strict subset of additive valuations.}, whose truthfulness also holds in expectation.
To the best of our knowledge, no non-trivial results are known for homogeneous divisible chores.

\paragraph{}
In light of the huge discrepancy between truthful good and chore allocations, we henceforth ask:

\begin{quote}
\em Are fairness and efficiency compatible with truthfulness for chores?
    In particular, when can we apply the tools from truthful good allocation to handle chores?    
\end{quote}
Regarding the latter question, we emphasize that allocating chores is known to be more challenging than goods in various aspects (see, e.g., \cite{DBLP:conf/sigecom/ChaudhuryGMM22,DBLP:conf/ijcai/GargMQ23}), and, to the best of our knowledge, no general connection between good and chore allocations is known.

\subsection{Our Contributions}

In this paper, we bridge the gap between truthful good and chore allocations by observing their connections in various important settings.
We assume agents' cost functions to be additive and describe our results as follows.

For indivisible items with two agents, it has been observed by several prior works~\cite{DBLP:journals/scw/BeiHS20,DBLP:journals/aamas/Segal-Halevi20,DBLP:journals/siamdm/BeiESS25} that the allocation of goods and that for chores can be equated by the following simple reduction: When allocating goods (resp. chores), we treat the utility (resp. cost) functions as cost (resp. utility) functions and allocate the items as if they are chores (resp. goods), and then we swap the bundles in the allocation.
We show that this reduction also preserves truthfulness and various fairness properties (\Cref{thm:trans-indiv-two-agents}).
This, together with the characterization by \cite{DBLP:conf/sigecom/AmanatidisBCM17} for goods, allows us to characterize the family of truthful mechanisms for chores with two agents (\Cref{thm:char-truth-chores-informal}) and derive the tight guarantees achieved by truthful mechanisms with chores for various fairness notions (\Cref{coro:indiv-ef1} and \Cref{coro:indiv-mms}).

Moving to homogeneous divisible items, we notice that the above reduction from chore allocation to good allocation with two agents can be generalized to an arbitrary number of agents.
To illustrate, suppose that we are given a mechanism with prominent properties for goods, and we invoke it with agents' cost functions as input to obtain an allocation.
Intuitively, in this allocation, each agent's bundle consists of items that are highly undesired by him, and hence, we uniformly distribute each agent's bundle among all other agents.
We show that such an operation transforms fairness and truthfulness for goods to the corresponding properties for chores (\Cref{thm:div-goods-chores}), which, combined with the results by \cite{freeman2023equivalence}, enables us to characterize the family of truthful mechanisms for chores with two agents subject to some mild technical conditions (\Cref{thm:divi-char-chores}) and show that the \emph{efficiency ratio}, the approximation ratio in \emph{social welfare}, of every truthful mechanism for chores with two agents is $0$ (\Cref{thm:pot-homo-chores}).
Moreover, the family of strictly truthful, PROP, and EF mechanisms for goods with an arbitrary number of agents given by \cite{freeman2023equivalence} can be easily adapted to mechanisms for chores with the same guarantees (see the remarks after \Cref{thm:div-goods-chores}).

Finally, we study indivisible items for an arbitrary number of agents with \emph{bi-valued} cost functions, i.e., the cost of each agent for each item is either $p$ or $q$ for some $p > q > 0$, and the bi-valued setting has received extensive attention in recent years \cite{DBLP:journals/tcs/AmanatidisBFHV21,DBLP:journals/tcs/GargM23,DBLP:conf/atal/EbadianP022,DBLP:conf/aaai/GargMQ22}.
Recall that a randomized mechanism satisfies a property \emph{ex-ante} if it holds in expectation, and satisfies a property \emph{ex-post} if it holds for every realized integral allocation.
We give a randomized mechanism that satisfies ex-ante truthfulness, ex-ante \emph{Pareto optimality (PO)}, ex-ante EF, and ex-post EF1 for chores (\Cref{thm:bivalue-chores}).
This strengthens the guarantees of mechanisms by prior works \cite{DBLP:conf/atal/EbadianP022,DBLP:conf/aaai/GargMQ22} for the same setting, which only satisfy ex-post \emph{fractional PO (fPO)} and ex-post EF1 for chores.\footnote{Note that ex-ante PO implies ex-post fPO.}
To achieve this, we leverage the mechanism for homogeneous divisible goods with bi-valued agents recently introduced by \cite{DBLP:journals/corr/abs-2407-13634} and the implementation scheme of \cite{DBLP:conf/wine/Aziz20}.
In particular, we largely exploit the \emph{balance} property of the allocations returned by both of the aforementioned components, and considerable efforts are needed to prove the promised guarantees of our mechanism.

\subsection{Other Related Work}

During the past decades, there have been plentiful results in the area of fair allocation.
In this subsection, we only focus on the most related works in the strategic setting and refer to the recent survey~\cite{Amanatidis_2023} for a more comprehensive overview of the fair allocation literature.

\paragraph{Truthful Cake Cutting.}
Besides the cases where resources are modeled as indivisible or homogeneous divisible items, another popular model assumes resources to be a heterogeneous divisible item, which is often referred to as \emph{cake cutting}.
The majority of literature on cake cutting assumes the cake to be a good.
The first truthful and EF cake-cutting mechanism for \emph{piecewise-uniform} valuations is given by \cite{DBLP:journals/geb/ChenLPP13}, which is later shown to be equivalent to the \emph{Maximum Nash Welfare} mechanism and satisfy a stronger notion of \emph{group strategyproofness} \cite{DBLP:conf/wine/AzizY14}.
For two agents with piecewise-uniform valuations, \cite{DBLP:journals/scw/BeiHS20} give a truthful and EF mechanism without the \emph{free-disposal} assumption and derive the same result when the cake is a chore by the reduction from chores allocation to goods allocation.
For the more general piecewise-constant valuations, \cite{DBLP:journals/ai/BuST23} show the non-existence of truthful and PROP cake-cutting mechanisms even for two agents.

\paragraph{Beyond Truthful Mechanisms.}
Due to the strong impossibility of simultaneously achieving truthfulness and fairness, several relaxations of truthfulness have been proposed.
\cite{DBLP:wine/AmanatidisBFLLR21,DBLP:conf/sigecom/AmanatidisBL0R23} study the fairness of the outcomes induced by \emph{pure Nash equilibria} with respect to the underlying true valuations, which they term as \emph{equilibrium fairness}.
A recent series of research~\cite{DBLP:journals/jcss/HuangWWZ24,tao2024fairtruthfulmechanismsadditive,bei2025incentive,DBLP:conf/atal/0037SX24} adopts the concept of \emph{incentive ratio}, which quantitatively limits each agent's gain by manipulation in the worst case.
Other relaxed variants of truthfulness include \emph{maximin strategyproofness}~\cite{brams2006better}, \emph{non-obvious manipulability}~\cite{DBLP:conf/nips/0001V22,ortega2022obvious}, and \emph{risk-averse truthfulness}~\cite{DBLP:journals/ai/BuST23,DBLP:journals/corr/abs-2502-18805}.

%% file: prelim.tex
\section{Preliminaries}

Let $N = [n]$ denote the set of agents and $O = \{o_1, \ldots, o_m\}$ denote the set of items.
We will assume that either all items are \emph{goods} or all of them are \emph{chores}, and use \emph{utilities} and \emph{costs} to respectively describe the values of items.
In particular, goods bring non-negative utilities, and agents are happier with a higher utility; conversely, chores bring non-negative costs, and agents are happier with a lower cost.
We define the settings of indivisible items and homogeneous divisible items separately.

\subsection{Indivisible Items}

For indivisible items, an \emph{(integral) allocation} $A = (A_1, \ldots, A_n)$ is a partition of $O$ satisfying $\bigcup_{i \in N} A_i = O$ and $A_i \cap A_j = \emptyset$ for all $i, j \in N$ with $i \neq j$, where $A_i$ is the bundle assigned to agent $i$.
When items are goods, each agent $i$ is associated with a utility function $v_i: 2^O \to \mathbb{R}_{\geq 0}$.
For notational simplicity, for all $i \in N$, $O' \subseteq O$, and $o \in O$, we will write $v_i(\{o\})$, $v_i(O' \cup \{o\})$, and $v_i(O' \setminus \{o\})$ as $v_i(o)$, $v_i(O' + o)$, and $v_i(O' - o)$, respectively.
We assume utility functions to be \emph{additive}, i.e., $v_i(S) = \sum_{o \in S} v_i(o)$ for all $i \in N$ and $S \subseteq O$.
We adopt the same notations and assumptions for the cost functions $c_i: 2^O \to \R_{\geq 0}$ when items are chores.
Given an allocation $A$, the \emph{utility} (resp. \emph{cost}) of each agent $i$ is defined as $v_i(A_i)$ (resp. $c_i(A_i)$).

\begin{definition}[EF1]
    An allocation $A$ is \emph{envy-free up to one item (EF1) for goods (resp. chores)} if for all $i, j \in N$ with $A_j \neq \emptyset$ (resp. $A_i \neq \emptyset$), there exists $o \in A_j$ (resp. $o \in A_i$) such that $v_i(A_i) \geq v_i(A_j - o)$ (resp. $c_i(A_i - o) \leq c_i(A_j)$).
\end{definition}

\begin{definition}[MMS]
    For each agent $i$, define his \emph{maximin share for goods (resp. chores)} as $\MMS_i^g
        = \max_A \min_{j \in N} v_i(A_j)$ (resp. $\MMS_i^c = \min_A \max_{j \in N} c_i(A_j)$).
    An allocation $A$ is \emph{$\alpha$-approximate maximin share fair ($\alpha$-MMS) for goods (resp. chores)} if $v_i(A_i) \geq \alpha \cdot \MMS_i^g$ (resp. $c_i(A_i) \leq \alpha \cdot \MMS_i^c$) for every $i \in N$.
    Note that $\alpha\leq 1$ for goods and $\alpha\geq 1$ for chores.
\end{definition}

\subsection{Homogeneous Divisible Items}

For homogeneous divisible items, a (fractional) \emph{allocation} 
$\bfx = (x_1, \ldots, x_n) \in [0, 1]^{n \times m}$ satisfies $\sum_{i \in N} x_i(o) = 1$ for every $o \in O$, where $x_i$ denotes the bundle received by agent $i$ and $x_i(o)$ denotes the fraction of item $o$ allocated to agent $i$.
When items are goods, each agent $i$ is associated with a non-negative utility function $v_i$ where $v_i(x_i)$ represents the utility of agent $i$ for every bundle $x_i$.
For every $S \subseteq O$, we slightly abuse notations and use $v_i(S)$ to denote $v_i(\mathbf{1}_S)$, where $\mathbf{1}_S \in [0, 1]^m$ is the bundle such that $\mathbf{1}_S(o) = 1$ for every $o \in S$ and $\mathbf{1}_S(o) = 0$ for every $o \in O \setminus S$.
For each $o \in O$, we use $v_i(o)$ to denote $v_i(\mathbf{1}_{\{o\}})$.
We assume utility functions to be \emph{additive}, i.e., $v_i(x_i) = \sum_{o \in O} x_i(o) \cdot v_i(o)$.
We normalize agents' utility functions so that $v_i(O) = 1$ for every $i \in N$.
We define the cost functions $c_i(\cdot)$ in the same manner for chores.
Given an allocation $\bfx$, the \emph{utility (resp. \emph{cost})} of each agent $i$ is defined as $v_i(x_i)$ (resp. $c_i(x_i)$).

\begin{definition}[EF]
    An allocation $\bfx$ is \emph{envy-free (EF) for goods (resp. chores)} if for all $i, j \in N$, $v_i(x_i) \geq v_i(x_j)$ (resp. $c_i(x_i) \leq c_i(x_j)$).
\end{definition}

\begin{definition}[PROP]
    An allocation $\bfx$ is \emph{proportional (PROP) for goods (resp. chores)} if for all $i \in N$, $v_i(x_i) \geq 1 / n$ (resp. $c_i(x_i) \leq 1 / n$).
\end{definition}

\begin{definition}[PO]
    An allocation $\bfx$ is \emph{Pareto optimal (PO) for goods (resp. chores)} if there is no allocation $\bfx'$ with $v_i(x_i') \geq v_i(x_i)$ (resp. $c_i(x_i') \leq c_i(x_i)$) and at least one inequality being strict.
\end{definition}

\subsection{Mechanisms}

When items are goods, a \emph{mechanism} $\calM$ takes as input a utility profile $\bfv = (v_1, \ldots, v_n)$ and outputs an allocation $\calM(\bfv)$, where we denote $\calM_i(\bfv)$ as the bundle received by agent $i$.
When items are chores, the input of $\calM$ is a cost profile $\bfc=(c_1, \ldots, c_n)$.
We define the truthfulness property of mechanisms for goods and chores separately, which requires that no agent can be strictly happier by misreporting his preference.

\begin{definition}[Truthfulness]
    A mechanism $\calM$ is \emph{truthful for goods (resp. chores)} if for every utility profile $\bfv = (v_1, \ldots, v_n)$ (resp. cost profile $\bfc = (c_1, \ldots, c_n)$), for all $i \in N$ and utility function $v_i'$ (resp. cost function $c_i'$), $v_i(\calM_i(\bfv)) \geq v_i(\calM_i(v_i', v_{-i}))$ (resp. $c_i(\calM_i(\bfc)) \leq c_i(\calM_i(c_i', c_{-i}))$).
    We say that $\calM$ is \emph{strictly truthful} if the above inequality is strict whenever $v_i' \neq v_i$.
\end{definition}

For a mechanism $\calM$ and a utility profile $\bfv$ (resp. cost profile $\bfc$), denote the optimal \emph{social welfare for goods (resp. chores)} as $\SW^g(\bfv) = \sum_{o \in O} \max_{i \in N} v_i(o)$ (resp. $\SW^c(\bfc) = \sum_{o \in O} \min_{i \in N} c_i(o)$), and denote $\SW(\calM(\bfv)) = \sum_{i \in N} v_i(\calM_i(\bfv))$ (resp. $\SW(\calM(\bfc)) = \sum_{i \in N} c_i(\calM_i(\bfc))$)
as the social welfare achieved by $\calM$ on a given profile.
Next, we define the \emph{efficiency ratio} of a truthful mechanism, which measures the worst-case approximation ratio of the social welfare achieved by the mechanism with respect to the optimal one.

\begin{definition}[Efficiency Ratio]
    The \emph{efficiency ratio for goods (resp. chores)} of a truthful mechanism $\calM$ is defined as
    \begin{align*}
        \inf_{\bfv} \frac{\SW(\calM(\bfv))}{\SW^g(\bfv)}
        \quad \left( \text{resp. } \inf_{\bfc} \frac{\SW^c(\bfc)}{\SW(\calM(\bfc)))} \right).
    \end{align*}
\end{definition}

We will also consider \emph{randomized mechanisms} for indivisible items.
We say that a randomized mechanism $\calM$ satisfies certain properties \emph{ex-ante} if the fractional allocation $\bfx$ defined by $x_i(o) = \Pr[o \in \calM_i(\cdot)]$ for all $i \in N$ and $o \in O$ satisfies such properties.
Moreover, we say that $\calM$ satisfies certain properties \emph{ex-post} if every integral allocation in its support satisfies such properties.

%% file: indiv-2-agent.tex
\section{Indivisible Items with Two Agents}

In this section, we assume that there are $n = 2$ agents, and items are indivisible.
We reduce the task of allocating chores to that of goods via a simple observation in prior works~\cite{DBLP:journals/scw/BeiHS20,DBLP:journals/aamas/Segal-Halevi20,DBLP:journals/siamdm/BeiESS25}: Given an arbitrary mechanism $\calM^g$ (possibly with desirable guarantees for goods), we define a new mechanism $\calM^c$ such that for every cost profile $\bfc$, $\calM^c(\bfc) = (\calM^g_2(\bfc), \calM^g_1(\bfc))$.
In other words, the allocation outputted by $\calM^c$ results from swapping the bundles in the allocation produced by $\calM^g$ under the same input.
Notice that for all cost profile $\bfc$ and $i \in N$, it holds that $\calM^c_i(\bfc) = O \setminus \calM^g_i(\bfc)$.

\begin{theorem}\label{thm:trans-indiv-two-agents}
    Let $\calM^g$ be a mechanism for indivisible items with two agents.
    The following statements hold:
    \begin{enumerate}
        \item $\calM^g$ is truthful for goods iff $\calM^c$ is truthful for chores.

        \item $\calM^g$ is EF1 for goods iff $\calM^c$ is EF1 for chores.

        \item For every $\alpha \in [0, 1]$, $\calM^g$ is $\alpha$-MMS for goods iff $\calM^c$ is $(2 - \alpha)$-MMS for chores.
    \end{enumerate}
\end{theorem}

\begin{proof}
    We only prove the direction that a certain property of $\calM^g$ for goods implies the corresponding property of $\calM^c$ for chores, and the opposite direction can be established analogously.
    Fix a cost profile $\bfc$, and we examine each property separately.
    
    \paragraph{Truthfulness.}
    Assuming that $\calM^g$ is truthful for goods, and agent $i$ manipulates his cost function as $c_i'$, it holds that
    \begin{align*}
        c_i(\calM_i^c(c_i', c_{-i}))
        &= c_i(O) - c_i(\calM^g_i(c_i', c_{-i}))\\
        &\geq c_i(O) - c_i(\calM^g_i(\bfc))
        = c_i(\calM^c_i(\bfc)),
    \end{align*}
    where the inequality holds by the truthfulness of $\calM^g$ for goods.
    Hence, $\calM^c$ is truthful for chores.

    \paragraph{EF1.}
    Suppose that $\calM^g$ is EF1 for goods, which implies that for every $i \in N$ with $\calM^g_{- i}(\bfc) \neq \emptyset$, there exists $o \in \calM^g_{- i}(\bfc)$ such that $c_i(\calM^g_i(\bfc)) \geq c_i(\calM^g_{- i}(\bfc) - o)$.
    Since $o \in \calM^g_{- i}(\bfc) = \calM^c_i(\bfc)$,
    \begin{align*}
        c_i(\calM_i^c(\bfc) - o)
        = c_i(\calM^g_{- i}(\bfc) - o)
        \leq c_i(\calM_i^g(\bfc))
        = c_i(\calM_{-i}^c(\bfc)),
    \end{align*}
    concluding that $\calM^c$ is EF1 for chores.

    \paragraph{$\alpha$-MMS.}
    Let $\MMS_i^g$ be the maximin share for goods with respect to the utility profile $\bfv := \bfc$.
    Suppose that $\calM^g$ is $\alpha$-MMS for goods, which implies that $c_i(\calM_i^g(\bfc))\geq \alpha \cdot \MMS_i^g$ for every $i \in N$.
    Note that $\MMS_i^g \leq \MMS_i^c$ and $\MMS_i^g+\MMS_i^c=c_i(O)$.
    As a result, for every $i \in N$,
    \begin{align*}
        \frac{c_i(\calM_i^c(\bfc))}{\MMS_i^c}=\frac{c_i(O)-c_i(\calM_i^g(\bfc))}{\MMS_i^c}
        \leq \frac{\MMS_i^c+(1-\alpha) \cdot \MMS_i^g}{\MMS_i^c}\leq 2-\alpha,
    \end{align*}
    concluding that $\calM^c$ is $(2 - \alpha)$-MMS for chores.
\end{proof}

We note here that one can similarly establish the connections for properties other than those listed in \Cref{thm:trans-indiv-two-agents} (e.g., EFX and PO), and we do not pursue this direction for the ease of presentation.
Moreover, we can analogously reduce the task of allocating goods to that of allocating chores while preserving various properties.

\subsection{Applications}

In this subsection, we discuss the applications of \Cref{thm:trans-indiv-two-agents}, primarily in the strategic setting.
Recall that the prominent characterization by \cite{DBLP:conf/sigecom/AmanatidisBCM17} states that the family of truthful mechanisms for indivisible goods with two agents coincides with the family of \emph{picking-exchange} mechanisms, which will be formally defined in \Cref{sec:indiv-char}.
By applying \Cref{thm:trans-indiv-two-agents}, we establish in \Cref{sec:indiv-char} the counterpart of the above characterization for chores, with the definition of picking-exchange mechanisms properly adapted.

\begin{theorem}[Informal]\label{thm:char-truth-chores-informal}
    For indivisible items with two agents, a mechanism is truthful for chores iff it is a picking-exchange mechanism for chores.
\end{theorem}

Next, we explore the compatibility between truthfulness and fairness for chores.
We start with the application to EF1 in the following corollary, which is a direct consequence of combining \cite[Application 4.6]{DBLP:conf/sigecom/AmanatidisBCM17} and \Cref{thm:trans-indiv-two-agents}.

\begin{corollary}\label{coro:indiv-ef1}
    Assume that items are indivisible and $n = 2$.
    For $m \leq 4$, there exists a truthful and EF1 mechanism for chores.
    Moreover, for $m \geq 5$, there is no truthful and EF1 mechanism for chores.
\end{corollary}

The application to MMS, which is stated in the following corollary, is a direct consequence of combining \cite[Application 4.7]{DBLP:conf/sigecom/AmanatidisBCM17} and \Cref{thm:trans-indiv-two-agents}.

\begin{corollary}\label{coro:indiv-mms}
    Assume that items are indivisible.
    For every $m$, there exists a truthful and $(2 - \lfloor \max\{2, m\} / 2 \rfloor^{-1})$-MMS mechanism for chores.
    Moreover, there is no truthful mechanism for chores that achieves a better MMS guarantee for chores.
\end{corollary}

In particular, Corollary~\ref{coro:indiv-mms} improves the upper bound of $4/3$ observed by \cite{DBLP:conf/ijcai/00010W19} for $n=2$ and $m=4$.

%% file: div.tex
\section{Homogeneous Divisible Items}
\label{sec:div}

In this section, we consider homogeneous divisible items and assume that there are $n \geq 2$ agents.
We first generalize the reduction given by the previous section to more than two agents.
Given an arbitrary mechanism $\calM^g$ (possibly with desirable guarantees for goods), we define a new mechanism $\calM^c$ as follows.
Given as input a cost profile $\bfc$, let $\bfx^g = \calM^g(\bfc)$, and the output of $\calM^c(\bfc)$, denoted as $\bfx^c$, satisfies $x_i^c(o) = \frac{1 - x_i^g(o)}{n - 1}$ for all $i \in N$ and $o \in O$.
To see that $\bfx^c$ defined above is a feasible allocation, notice that for all $i \in N$ and $o \in O$, $x_i^c(o) \in [0, 1]$ since $x_i^g(o) \in [0, 1]$.
Moreover, for every $o \in O$,
\begin{align*}
    \sum_{i \in N} x^c_i(o)
    = \sum_{i \in N} \frac{1 - x_i^g(o)}{n - 1}
    = \frac{n - \sum_{i \in N} x_i^g(o)}{n - 1}
    = 1,
\end{align*}
concluding that $\bfx^c$ is feasible.

\begin{theorem}\label{thm:div-goods-chores}
    Let $\calM^g$ be a mechanism for homogeneous divisible items.
    The following statements hold:
    \begin{enumerate}
        \item If $\calM^g$ is truthful for goods, then $\calM^c$ is truthful for chores.

        \item If $\calM^g$ is EF for goods, then $\calM^c$ is EF for chores.

        \item If $\calM^g$ is PROP for goods, then $\calM^c$ is PROP for chores.
    \end{enumerate}
\end{theorem}

\begin{proof}
    Fix a cost profile $\bfc$, and we examine each property separately.
    \paragraph{Truthfulness.}
    Assuming that $\calM^g$ is truthful for goods, and agent $i$ manipulates his cost function as $c_i'$, it holds that
    \begin{align*}
        c_i(\calM_i^c(c_i', c_{-i}))
        &= \frac{c_i(O) - c_i(\calM_i^g(c_i', c_{-i}))}{n - 1}  \\
        &\geq \frac{c_i(O) - c_i(\calM_i^g(\bfc))}{n - 1}
        = c_i(\calM^c_i(\bfc)),
    \end{align*}
    where the inequality holds due to the truthfulness of $\calM^g$ for goods.
    Hence, $\calM^c$ is truthful for chores.

    \paragraph{EF.}
    Suppose that $\calM^g$ is EF for goods, which implies that $c_i(\calM^g_i(\bfc)) \geq c_i(\calM^g_j(\bfc))$ for all $i, j \in N$.
    As a result,
    \begin{align*}
        c_i(\calM_i^c(\bfc))
        = \frac{c_i(O) - c_i(\calM_i^g(\bfc))}{n - 1}
        \leq \frac{c_i(O) - c_i(\calM^g_j(\bfc))}{n - 1}
        = c_i(\calM_j^c(\bfc)),
    \end{align*}
    concluding that $\calM^c$ is EF for chores.

    \paragraph{PROP.}
    Suppose that $\calM^g$ is PROP for goods, which implies that $c_i(\calM_i^g(\bfc)) \geq c_i(O) / n$ for every $i \in N$.
    As a result,
    \begin{align*}
        c_i(\calM_i^c(\bfc))
        = \frac{c_i(O) - c_i(\calM_i^g(\bfc))}{n - 1}
        \leq \frac{c_i(O) - c_i(O) / n}{n - 1}
        = \frac{c_i(O)}{n},
    \end{align*}
    concluding that $\calM^c$ is PROP for chores.
\end{proof}

We remark here that the transformation provided in this section does not preserve PO when $n \geq 3$.
To illustrate, suppose that there are $n = 3$ agents and $m = 2$ items, and the cost functions satisfy $c_1(o_1) = c_2(o_1) = c_3(o_2) = 1$ and $c_1(o_2) = c_2(o_2) = c_3(o_1) = 0$.
Let $\calM^g$ be a mechanism that allocates $o_1$ to agent $1$ and $o_2$ to agent $3$, which results in a PO allocation for goods.
However, $\calM^c$ will allocate half of $o_1$ to agents $2$ and $3$, respectively, which is not PO for chores since allocating the entire $o_1$ to agent $3$ instead would lead to a Pareto improvement.

We mention a few direct applications of \Cref{thm:div-goods-chores}.
Recall that \cite{freeman2023equivalence} establish an equivalence between the class of fair allocation mechanisms and the class of \emph{weakly budget-balanced} wagering mechanisms, which results in a family of (strictly) truthful, PROP, and EF mechanisms for goods.
Moreover, \cite{DBLP:conf/wine/AzizY14} propose the \emph{Constrained Serial Dictatorship} mechanism, which is truthful and PROP for goods.
Notice that the proof of truthfulness in \Cref{thm:div-goods-chores} can be easily extended to show that strict truthfulness is also preserved.
Therefore, we can apply \Cref{thm:div-goods-chores} to the aforementioned mechanisms, leading to new mechanisms with the same (strict) truthfulness and fairness guarantees for chores.

\subsection{Characterization}

In this subsection, we restrict our attention to the case of $n = 2$ agents.
Recently, \cite{freeman2023equivalence} prove that all truthful mechanisms for goods are \emph{swap-dictatorial} subject to mild technical conditions.
With the aid of \Cref{thm:div-goods-chores}, we generalize the above characterization to chores.

We start by introducing some necessary definitions.

\begin{definition}
    A mechanism $\calM$ is \emph{anonymous} if for every utility profile $\bfv = (v_1, v_2)$, $\calM_1(v_1, v_2) = \calM_2(v_2, v_1)$ and $\calM_2(v_1, v_2) = \calM_1(v_2, v_1)$.
    The same definition holds for cost profiles.
\end{definition}

\begin{definition}
    A \emph{swap-dictatorial} mechanism $\calM$ for homogeneous divisible goods (resp. chores) is defined by a (potentially infinite) set of bundles $D \subseteq [0, 1]^m$ such that for every valuation profile $\bfv$ (resp. cost profile $\bfc$), $\calM_1(\bfv) = (x_1 + \mathbf{1}_O - x_2) / 2$ (resp. $\calM_1(\bfc) = (x_1 + \mathbf{1}_O - x_2) / 2$), where $x_1 \in \argmax_{x \in D} v_1(x)$ (resp. $x_1 \in \argmin_{x \in D} c_1(x)$) and $x_2 \in \argmax_{x \in D} v_2(x)$ (resp. $x_2 \in \argmin_{x \in D} c_2(x)$) with ties broken arbitrarily.
\end{definition}

As an intuitive interpretation of a swap-dictatorial mechanism, each agent has an equal probability of being selected as the dictator, who is allowed to receive his favorite bundle among $D$, with the other agent receiving the remaining items, and the outcome of the mechanism is the expectation over two possible choices of the dictator.
We emphasize here that although we give an interpretation with a randomized implementation, every swap-dictatorial mechanism for goods (resp. chores) is deterministic, and it is easy to verify by definition that it further satisfies truthfulness for goods (resp. chores).

\begin{definition}
        A mechanism $\calM$ for homogeneous divisible items is \emph{smooth} if, for every $i \in N$, $\calM_i(\bfv)$ is twice continuously differentiable with respect to every utility function $v_j$, $j \in N$.\footnote{Here, we encode a utility function $v_j$ by the vector $(v_j(o_1), \ldots, v_j(o_m)) \in [0, 1]^m$, and we say that a condition holds with respect to a utility function if it holds with respect to the corresponding vector in $[0, 1]^m$.}
        The same definition holds for cost profiles.
\end{definition}

With the above definitions, we are ready to present the characterization of truthful mechanisms for goods by \cite{freeman2023equivalence}.

\begin{theorem}[Theorem 6 in \cite{freeman2023equivalence}]\label{thm:divi-char-goods}
    For homogeneous divisible items with two agents, an anonymous and smooth mechanism is truthful for goods iff it is swap-dictatorial for goods.
\end{theorem}

Next, we provide the characterization of truthful mechanisms for chores by combining \Cref{thm:div-goods-chores} and \Cref{thm:divi-char-goods}.

\begin{theorem}\label{thm:divi-char-chores}
    For homogeneous divisible items with two agents, an anonymous and smooth mechanism is truthful for chores iff it is swap-dictatorial for chores.
\end{theorem}

\begin{proof}
    According to previous discussions, every swap-dictatorial mechanism for chores is truthful for chores, and it remains to show that every anonymous, smooth, and truthful mechanism for chores is swap-dictatorial for chores.
    Let $\calM^c$ be an anonymous, smooth, and truthful mechanism for chores, and let $\calM^g$ be the mechanism satisfying $\calM^g(\bfc) = (\calM^c_2(\bfc), \calM^c_1(\bfc))$ for every cost profile $\bfc$.
    Notice that $\calM^g$ is anonymous and smooth, and $\calM^g$ further satisfies truthfulness for goods by \Cref{thm:div-goods-chores}.
    As a result, $\calM^g$ is swap-dictatorial for goods by \Cref{thm:divi-char-goods}, which implies that there exists a set of bundles $D \subseteq [0, 1]^m$ such that for every cost profile $\bfc = (c_1, c_2)$, $\calM_1^g(\bfc) = (x_1 + \mathbf{1}_O - x_2) / 2$ with $x_1 \in \argmax_{x \in D} c_1(x)$ and $x_2 \in \argmax_{x \in D} c_2(x)$.
    Let $D' = \{\mathbf{1}_O - x \mid x \in D\}$, and it holds that
    \begin{align*}
        \calM_1^c(\bfc)
        &= \mathbf{1}_O - \calM_1^g(\bfc)
        = \mathbf{1}_O - \frac{x_1 + \mathbf{1}_O - x_2}{2}\\
        &= \frac{(\mathbf{1}_O - x_1) + \mathbf{1}_O - (\mathbf{1}_O - x_2)}{2}.
    \end{align*}
    Notice that $\mathbf{1}_O - x_1 \in \argmin_{x \in D'} c_1(x)$ by the definition of $D'$ and the fact that $x_1 \in \argmax_{x \in D} c_1(x)$, and, similarly, $\mathbf{1}_O - x_2 \in \argmin_{x \in D'} c_2(x)$, implying that $\calM^c$ is swap-dictatorial for chores.
\end{proof}

\subsection{Efficiency Ratio}

In this subsection, we assume that there are $n = 2$ agents and focus on truthful mechanisms.
It is easy to show that allocating each item equally to agents achieves an efficiency ratio of $0.5$ for goods.
Moreover, \cite{DBLP:conf/atal/GuoC10} show that the family of \emph{increasing-price} mechanisms cannot achieve an efficiency ratio better than $0.5$ for goods, and the negative result is later generalized to all swap-dictatorial mechanisms by \cite{DBLP:conf/wine/HanSTZ11}.
The characterization by \cite{freeman2023equivalence} (\Cref{thm:divi-char-goods}), together with the negative result of \cite{DBLP:conf/wine/HanSTZ11}, further confirms $0.5$ to be the best efficiency ratio achievable for goods (subject to smoothness).\footnote{Mechanisms with an efficiency ratio better than $0.5$ exist when allowing some items to remain unallocated \cite{DBLP:conf/atal/ColeGG13,DBLP:conf/ijcai/Cheung16}.}
We apply the characterization by \Cref{thm:divi-char-chores} and the hard instances in the proof of \cite[Theorem 4]{DBLP:conf/wine/HanSTZ11} to show that every truthful mechanism for chores has an efficiency ratio of $0$ for chores subject to smoothness.

\begin{theorem}\label{thm:pot-homo-chores}
    For homogeneous divisible items with two agents, if a mechanism $\calM$ is smooth and truthful for chores, then $\calM$ has an efficiency ratio of $0$ for chores.
\end{theorem}

It is worth mentioning that \cite{DBLP:journals/scw/BeiHS20} show that in \emph{cake-cutting} with two agents, the efficiency ratio of every EF and truthful mechanism for chores has an efficiency ratio of $0$ for chores.
However, \Cref{thm:pot-homo-chores} is arguably stronger than their result since (1) agents have less room to manipulate for homogeneous divisible items than in cake-cutting, and (2) \Cref{thm:pot-homo-chores} holds even for non-EF mechanisms.

\begin{proof}[Proof of \Cref{thm:pot-homo-chores}]
    We say that a mechanism $\calM$ is \emph{item-symmetric} if when the values of two items are swapped by every agent, the allocations of these two items are also swapped.
Recall that \cite[Claim 1]{DBLP:conf/atal/GuoC10} shows that for every truthful mechanism with an efficiency ratio of $\delta$ for goods, there exists a corresponding anonymous, item-symmetric, and truthful mechanism with an efficiency ratio of at least $\delta$ for goods.
We prove a counterpart of this statement for chores in the following lemma.
Throughout the proof, we use $\Sym(O)$ to denote the set of all permutations of $O$.
Moreover, given $\sigma \in \Sym(O)$ and a cost function $c_i$, let $c_i^{\sigma}$ denote the cost function satisfying $c_i^{\sigma}(\sigma(o)) = c_i(o)$ for every $o \in O$.

\begin{lemma}\label{lem:anony-item-symme}
    For homogeneous divisible items with two agents, given a truthful mechanism with an efficiency ratio of $\delta$ for chores, there exists a corresponding anonymous, item-symmetric, and truthful mechanism with an efficiency ratio of at least $\delta$ for chores.
\end{lemma}

The proof of \Cref{lem:anony-item-symme} is deferred to \Cref{sec:proof-div}.
By \Cref{lem:anony-item-symme}, it suffices to show that every anonymous, item-symmetric, smooth, and truthful mechanism for chores has an efficiency ratio of $0$ for chores.
Assume for contradiction that there exists an anonymous, item-symmetric, smooth, and truthful mechanism $\calM$ with an efficiency ratio of $\delta > 0$ for chores, and $\calM$ must be swap-dictatorial for chores by \Cref{thm:divi-char-chores}.
Let $D$ be the dictator's choice space, and let $m^{(1)} = 2^k$ and $m^{(i + 1)} = m^{(i)} / 2$ for $i \in [k - 1]$ with $k$ to be specified at the end of the proof.
Fix an arbitrary real number $t$ with $0 < t < \delta / (2 - \delta) \leq 1$, and we define a series of instances $\{\bfc^{(i)}\}_{i \in [k]}$.
In particular, agents' cost functions in instance $\bfc^{(i)}$ is defined as
\begin{align*}
    c_1^{(i)}(o) =
    \begin{cases}
        p, & o \in O^i_1,\\
        q, & o \in O^i_2,\\
        0, & \text{otherwise},
    \end{cases}
    \quad \text{and} \quad
    c_2^{(i)}(o) =
    \begin{cases}
        q, & o \in O^{(i)}_1,\\
        p, & o \in O^{(i)}_2,\\
        0, & \text{otherwise},
    \end{cases}
\end{align*}
where $O^{(i)}_1 = \{o_1, \ldots, o_{m^{(i)} / 2}\}$, $O^{(i)}_2 = \{o_{m^{(i)} / 2 + 1}, \ldots, o_{m^{(i)}}\}$, and $q / p = t$.
Suppose that when agent $1$ acts as the dictator on $\bfc^{(i)}$, he chooses the bundle $x^{(i)} \in D$, and let
\begin{align*}
    a^{(i)} = \frac{2}{m^{(i)}} \sum_{o \in O_1^{(i)}} x^{(i)}(o)
    \quad \text{and} \quad
    b^{(i)} = \frac{2}{m^{(i)}} \sum_{o \in O_2^{(i)}} x^{(i)}(o)
\end{align*}
denote the fractions of $O^{(i)}_1$ and $O^{(i)}_2$ contained in $x^{(i)}$, respectively.
We will show that as $i$ increases, $a^{(i)}$ decreases by a relative amount in order to maintain the efficiency ratio for chores.
However, $a^{(i)}$ cannot be smaller than $0$, which results in a contradiction.

It is known that a swap-dictatorial mechanism (for either goods or chores) that is anonymous and item-symmetric satisfies the following extra property.

\begin{lemma}[\cite{DBLP:conf/wine/HanSTZ11}]\label{lem:symmetric-SD}
    Assume that items are divisible and homogeneous.
    Let $D$ be the dictator's choice space of a swap-dictatorial mechanism (for either goods or chores) that is anonymous and item-symmetric.
    If a bundle $x = (x(o_1), \ldots, x(o_m)) \in D$, then $(x(\sigma(o_1)), \ldots, x(\sigma(o_m))) \in D$ for every $\sigma \in \Sym(O)$.
\end{lemma}

By \Cref{lem:symmetric-SD} and the anonymity and item-symmetry of $\calM$, it follows that agent $2$ takes $a^{(i)}$ fraction of $O_2^{(i)}$ and $b^{(i)}$ fraction of $O_1^{(i)}$ when he becomes the dictator on $\bfc^{(i)}$.
Since the efficiency ratio of $\calM$ for chores is $\delta$, the social welfare achieved by $\calM$ on $\bfc^{(i)}$ satisfies
\begin{align*}
    \delta \cdot \left(p \cdot a^{(i)} + q \cdot b^{(i)} + q \cdot (1 - a^{(i)}) + p \cdot (1 - b^{(i)}) \right) \cdot \frac{m^{(i)}}{2} \cdot \frac{1}{2} \cdot 2
    \leq \SW^c(c_1^{(i)}, c_2^{(i)})
    = q \cdot m^{(i)}.
\end{align*}
Rearranging the inequalities, we obtain
\begin{align}\label{eqn:pot}
    a^{(i)} - b^{(i)} \leq \frac{2t - \delta(1 + t)}{\delta(1 - t)}.
\end{align}

On the other hand, for $\bfc^{(i)}$, since agent $1$ chooses $x^{(i)}$ from $D$ when he becomes the dictator, the resulting cost for agent $1$ must be no larger than that obtained from choosing $x^{(i - 1)}$, as has been chosen for $\bfc^{(i - 1)}$.
By \Cref{lem:symmetric-SD} and the anonymity and item-symmetry of $\calM$, there exists a permutation of $x^{(i - 1)}$ in $D$ such that the average of the first $m^{(i)} / 2$ components is no more than the average of the second $m^{(i)} / 2$ components, and denote such an allocation by $x$.
By comparing agent $1$'s cost on $\bfc^{(i)}$ between choosing $x^{(i)}$ and $x$, it holds that
\begin{align*}
    (p \cdot a^{(i)} + q \cdot b^{(i)}) \cdot \frac{m^{(i)}}{2} \leq (p + q) \cdot a^{(i - 1)} \cdot \frac{m^{(i)}}{2}.
\end{align*}
By rearranging the terms, it follows that
\begin{align}\label{eqn:dictate}
    a^{(i)} \cdot \frac{1}{t + 1} + b^{(i)} \cdot \frac{t}{t + 1} \leq a^{(i - 1)}.
\end{align}

Finally, combining \eqref{eqn:pot} and \eqref{eqn:dictate}, we obtain
\begin{align*}
    a^{(i)}
    \leq a^{(i - 1)} + \frac{t}{t + 1} \cdot \frac{2t - \delta(1 + t)}{\delta(1 - t)}.
\end{align*}
Noticing that $a^{(1)} \leq 1$ leads to
\begin{align*}
    a^{(k)}
    \leq 1 + (k - 1) \cdot \frac{t}{t + 1} \cdot \frac{2t - \delta(1 + t)}{\delta(1 - t)}.
\end{align*}
By the assumption that $0 < t < \delta / (2 - \delta) \leq 1$, we know that $2t - \delta(1 + t) < 0$.
Therefore, by choosing a sufficiently large $k$, we derive that $a^{(k)} < 0$, contradicting the fact that $a^{(k)} \in [0, 1]$.
\end{proof}

%% file: bobw.tex
\section{Indivisible Items with Bi-Valued Preferences}
\label{sec:bobw}

In this section, we assume that there are $n \geq 2$ agents, and we give an ex-ante truthful, ex-ante PO, ex-ante EF, and ex-post EF1 mechanism for indivisible chores with \emph{bi-valued} cost functions.
In particular, we say that the cost functions are bi-valued if there exists $p > q > 0$\footnote{When $q = 0$, the case of bi-valued cost functions are equivalent to that of binary cost functions, which has been thoroughly studied by prior work, e.g., \cite{DBLP:journals/eor/SunC25}.} such that $c_i(o) \in \{p, q\}$ for all $i \in N$ and $o \in O$.
Recall that \cite{DBLP:journals/eor/SunC25} give a mechanism for \emph{restricted additive} cost functions\footnote{We say that the cost functions are \emph{restricted additive} if each item $o$ is associated with an intrinsic value $e_o$ such that $c_i(o) \in \{0, e_o\}$ for every agent $i$.} satisfying numerous ex-ante and ex-post properties for chores, which is, to the best of our knowledge, the only known result of this kind for chores.
Since the set of bi-valued instances and the set of restricted additive instances are not a subset of each other, our result is incomparable with that of \cite{DBLP:journals/eor/SunC25}.

\begin{theorem}\label{thm:bivalue-chores}
    Suppose that agents have bi-valued cost functions and items are indivisible.
    Then, there exists an ex-ante truthful, ex-ante PO, ex-ante EF, and ex-post EF1 mechanism for chores.
\end{theorem}

As a remark, it is known that the integral allocations returned by an ex-ante PO mechanism satisfy ex-post \emph{fractional PO}, which is a stronger notion of ex-post PO \cite{aziz2023best}.
Moreover, we skip the discussion on time complexity throughout as it is not our main concern, yet our mechanism can be implemented in polynomial time.

\begin{proof}[Proof of \Cref{thm:bivalue-chores}]
Recall that the \emph{probabilistic serial (PS)} rule is a mechanism for divisible items in which all agents simultaneously consume their favorite items among those that have not been fully allocated at a unit speed, assuming that all items have the same ``size''.
Notice that the final outcome depends on the tie-breakings when an agent exhibits identical preferences on multiple items, and we allow an agent to switch to consuming another item before the current consumed item becomes unavailable.
We refer to \cite{DBLP:conf/wine/AzizY14,DBLP:conf/atal/0037SX24} for a more formal description of the PS rule.
Since the definition of favorite items relies on the nature of items, i.e., whether they are goods or chores, we use \PSg~to denote the PS rule where agents prefer items with higher utilities, and \PSc~to denote the PS rule where agents prefer items with lower costs.

We first present a truthful, PO, and EF mechanism for chores in the following lemma, assuming that cost functions are bi-valued and items are divisible, whose proof is deferred to \Cref{sec:proof-divi-chores}.
For each fractional bundle $x$, we use $|x| = \sum_{o \in O} x(o)$ to denote the size of $x$.

\begin{lemma}\label{lem:divi-truth-po-chores}
    When cost functions are bi-valued and items are divisible, there exists a truthful, PO, and EF mechanism $\calM$ for chores such that every fractional allocation $\bfx$ outputted by $\calM$ is an outcome of \PSc and satisfies $|x_i| = m / n$ for every $i \in N$.
\end{lemma}

Next, we show that every fractional allocation outputted by \PSc~can be \emph{implemented} over a set of integral allocations (see the formal description in \Cref{lem:decomp-ps-plus}) satisfying certain structural properties.
We achieve this by applying the machinery of \cite{DBLP:conf/wine/Aziz20} in a black-box manner and strengthen the guarantees of the resulting integral allocations originally proved by \cite{DBLP:conf/wine/Aziz20}, which is crucial for us to establish the ex-post EF1 property for chores.

\begin{lemma}\label{lem:decomp-ps-plus}
    Suppose that agents have additive cost functions.
    For every fractional allocation $\bfx$ corresponding to an outcome of \PSc, we can find a random integral allocation $\widehat{A}$ such that $\Pr[o \in \widehat{A}_i] = x_{i}(o)$ for all $i \in N$ and $o \in O$.
    Moreover, for each integral allocation $A$ in the support of $\widehat{A}$, it satisfies that $|A_i| \leq |A_j| + 1$ for all $i, j \in N$, and we can label the items in each bundle $A_i$ as $o_i^1, o_i^2, \ldots, o_i^{|A_i|}$ such that
    \begin{enumerate}
        \item \label{item:con1} $c_i(o_i^k) \leq c_i(o_i^{k + 1})$ for all $i \in N$ and $k \in [|A_i| - 1]$,
        \item \label{item:con2} for all $i, j \in N$ with $|A_i| \leq |A_j|$, for every $k \in [|A_j| - 1]$, $c_i(o_i^k) \leq c_i(o_j^{k + 1})$, and
        \item \label{item:con3} for all $i, j \in N$ with $|A_i| > |A_j|$, for every $k \in [|A_j|]$, $c_i(o_i^k) \leq c_i(o_j^k)$.
    \end{enumerate}
\end{lemma}

We defer the proof of \Cref{lem:decomp-ps-plus} to \Cref{sec:proof-decom-ps}.
Now, we are ready to finish the proof of \Cref{thm:bivalue-chores}.
    Let $\calM$ be the mechanism as per \Cref{lem:divi-truth-po-chores}, and we describe our mechanism $\calM'$ as follows.
    Recall that $\calM$ outputs fractional allocations and $\calM'$ outputs random integral allocations.
    Given as input a cost profile $\bfc$, we apply \Cref{lem:decomp-ps-plus} with $\bfx = \calM(\bfc)$ to obtain a random integral allocation $\widehat{A}$, which will be the output of $\calM'$.
    The ex-ante properties of $\calM'$ directly follow from the properties of $\calM$ promised by \Cref{lem:divi-truth-po-chores}, and it remains to prove that $\calM'$ satisfies ex-post EF1 for chores.
    
    Let $A$ be an integral allocation in the support of $\calM'(\bfc)$, which is obtained by applying \Cref{lem:decomp-ps-plus} with $\bfx = \calM(\bfc)$.
    By the guarantee of \Cref{lem:decomp-ps-plus}, $|A_i| \leq |A_j| + 1$ for all $i, j \in N$, and we label the items in each bundle $A_i$ by $o_i^1, o_i^2, \ldots, o_i^{|A_i|}$ as specified by \Cref{lem:decomp-ps-plus}.
    Next, we show that $A$ is EF1 for chores by analyzing the envy of agent $i$ toward agent $j$ depending on whether $|A_i| \leq |A_j|$.

    \paragraph{Case 1: $|A_i| \leq |A_j|$.}
    By \Cref{lem:decomp-ps-plus}, for every $k \in [|A_j| - 1]$, $c_i(o_i^k) \leq c_i(o_j^{k + 1})$.
    As a result,
    \begin{align*}
        c_i(A_i)
        &= \sum_{k=1}^{|A_i|} c_i(o_i^k)
        \leq c_i(o_i^{|A_i|}) + \sum_{k=1}^{|A_i| - 1} c_i(o_j^{k + 1})\\
        &\leq c_i(o_i^{|A_i|}) + c_i(A_j),
    \end{align*}
    which implies that agent $i$ does not envy agent $j$ up to one chore.

    \paragraph{Case 2: $|A_i| > |A_j|$.}
    By \Cref{lem:decomp-ps-plus}, for every $k \in [|A_j|]$, $c_i(o_i^k) \leq c_i(o_j^k)$.
    As a result,
    \begin{align*}
        c_i(A_i)
        &= \sum_{k=1}^{|A_i|} c_i(o_i^k)
        \leq c_i(o_i^{|A_i|}) + \sum_{k=1}^{|A_j|} c_i(o_j^k)\\
        &= c_i(o_i^{|A_i|}) + c_i(A_j),
    \end{align*}
    which implies that agent $i$ does not envy agent $j$ up to one chore.
\end{proof}

\subsection{Proof of \Cref{lem:divi-truth-po-chores}}
\label{sec:proof-divi-chores}

We first recall the truthful, PO, and EF mechanism of \cite{DBLP:journals/corr/abs-2407-13634} for divisible goods with bi-valued utility functions, denoted as $\calM^g$.
Given a bi-valued utility profile $\bfv$ as input, where $v_i(o) \in \{p, q\}$ for all $i \in N$ and $o \in O$ with $p > q > 0$, we assume that for each agent $i$, there exists an item $o$ such that $v_i(o) = p$, which is without loss of generality since if $v_i(o) = q$ for every item $o$, then we can equivalently treat the utility function of agent $i$ to satisfy $v_i(o) = p$ for every item $o$.
Let $\bfx'$ be an arbitrary (possibly partial) allocation that maximizes
\begin{align}\label{eqn:nash-welfare}
    f_{\bfv}(\bfx) = \prod_{i \in N} \sum_{o \in O: v_i(o) = p} x_{i}(o).
\end{align}
In particular, let $H_{\bfv} = \{o \in O \mid v_i(o) = q \text{ for every } i \in N\}$, and here we do not allocate the items in $H_{\bfv}$, i.e., $x_i'(o) = 0$ for all $i \in N$ and $o \in H_{\bfv}$.
Notice that $x'_i(o) > 0$ implies $v_i(o) = p$ for all $i \in N$ and $o \in O$ by the optimality of $\bfx'$, and $f_{\bfv}(\bfx') > 0$ since there exists at least one item $o$ for each agent $i$ such that $v_i(o) = p$.
Let $L = m / n$, and we prove a structural property of $\bfx'$ in the following lemma, which will be useful in the proof of \Cref{lem:bobw-po} to construct a market equilibrium.

\begin{lemma}\label{lem:char-mnw}
    For every $o \in O \setminus H_{\bfv}$, let $N_o = \{i \in N \mid x_{i}'(o) > 0\}$ be the set of agents receiving some fraction of $o$ in $\bfx'$.
    Then, either $|x_i'| > L$ for every $i \in N_o$, or $|x_i'| \leq L$ for every $i \in N_o$.
    Moreover, if $|x_i'| > L$ for some $i \in N_o$, then $v_j(o) = q$ for every $j \in N$ such that $|x_j'| \leq L$.
\end{lemma}

\begin{proof}
    Assume for contradiction that there exist $i, j \in N_o$ such that $|x_i'| > L$ and $|x_j'| \leq L$.
    Since $i, j \in N_o$, it follows that $v_i(o) = v_j(o) = p$.
    As a result, by moving $\min\{(|x_i'| - L) / 2, x_i'(o)\}$ fraction of $o$ from $x_i'$ to $x_j'$, we obtain from $\bfx'$ an allocation $\widehat{\bfx}$ with $f_{\bfv}(\widehat{\bfx}) > f_{\bfv}(\bfx')$, contradicting with the optimality of $\bfx'$.
    This concludes the first statement in \Cref{lem:char-mnw}.

    For the second statement in \Cref{lem:char-mnw}, assume for contradiction that there exists $i \in N_o$ and $j \in N$ such that $|x_i'| > L$, $|x_j'| \leq L$, and $v_j(o) = p$.
    Then, by moving $\min\{ (|x_i'| - L) / 2, x_i'(o)\}$ fraction of $o$ from $x_i'$ to $x_j'$, we obtain from $\bfx'$ an allocation $\widehat{\bfx}$ with $f_{\bfv}(\widehat{\bfx}) > f_{\bfv}(\bfx')$, contradicting with the optimality of $\bfx'$.
\end{proof}

We define another (possibly partial) fractional allocation $\bfx''$ as follows: for each agent $i$, if $|x_i'| \leq L$, then $x_i'' = x_i'$; otherwise, let $x_i''$ be an arbitrary subset of $x_i'$ with $|x_i''| = L$.
For each item $o$, notice that $\beta_o := 1 - \sum_{i \in N} x''_{i}(o)$ fraction of $o$ is unassigned in $\bfx''$, and our final allocation $\bfx$ satisfies
\begin{align*}
    x_{i}(o)
    = x_{i}''(o) + \beta_o \cdot \frac{L - |x_i''|}{\sum_{j \in N} (L - |x_j''|)},
\end{align*}
i.e., we allocate the unassigned part of each item to agents with the fraction received by each agent $i$ proportional to $L - |x_i''| \geq 0$.
Finally, $\calM^g$ outputs $\bfx$, i.e., $\calM^g(\bfv) = \bfx$.

\begin{lemma}[Theorem 6.1 in \cite{DBLP:journals/corr/abs-2407-13634}]\label{lem:bivalue-good-divi}
    When agents have bi-valued utility functions and items are divisible, $\calM^g$ is truthful, PO, and EF for goods.
    Moreover, every fractional allocation outputted by $\calM^g$ is an outcome of \PSg.
\end{lemma}

Next, we describe our mechanism $\calM^c$ for chores as follows.
Given a bi-valued cost profile $\bfc$ as input, where $c_i(o) \in \{p, q\}$ for all $i \in N$ and $o \in O$ with $p > q > 0$, let $\bfv$ be the utility profile satisfying $v_i(o) = (p + q) - c_i(o)$ for all $i \in N$ and $o \in O$, and $\bfv$ is also bi-valued.
Then, $\calM^c$ will output $\calM^c(\bfc) = \calM^g(\bfv)$.

Now, we aim to show that $\calM^c$ fulfills all the requirements.
As a crucial property, by the construction of $\bfv$, for every fractional bundle $x$,
\begin{align}\label{eqn:balance-value}
    c_i(x)
    &= \sum_{o \in O} c_i(o) \cdot x(o)
    = \sum_{o \in O} (p + q - v_i(o)) \cdot x(o)\\  \nonumber
    &= (p + q) \cdot |x| - v_i(x).
\end{align}
Due to the same reason as before, assume that for each agent $i$, there exists an item $o$ such that $c_i(o) = q$.
Since $\calM^g(\bfv)$ is an outcome of \PSg, $\calM^c(\bfc)$ is an outcome of \PSc~by the construction of $\bfv$ and the fact that $\calM^c(\bfc) = \calM^g(\bfv)$.

We first show that $\calM^c$ is truthful for chores.
Suppose that agent $i$ manipulates his cost function as $c_i'$, and let $v_i'$ be the utility function satisfying $v_i'(o) = (p + q) - c_i'(o)$ for every $o \in O$.
It follows that
\begin{align*}
    c_i(\calM_i^c(c_i', c_{-i}))
    &= c_i(\calM_i^g(v_i', v_{-i}))\\
    &\overset{(a)}{=} (p + q) \cdot |\calM_i^g(v_i', v_{-i})| - v_i(\calM_i^g(v_i', v_{-i}))\\
    &\overset{(b)}{\geq} (p + q) \cdot |\calM_i^g(\bfv)| - v_i (\calM_i^g(\bfv))\\
    &\overset{(c)}{=} c_i(\calM_i^g(\bfv))
    = c_i(\calM_i^c(\bfc)),
\end{align*}
where $(a)$ and $(c)$ holds by \eqref{eqn:balance-value}, and $(b)$ holds due to the truthfulness of $\calM^g$ for goods and the fact that $|\calM_i^g(v_i', v_{-i})| = |\calM_i^g(\bfv)| = L$.
This concludes that $\calM^c$ is truthful for chores.

Moreover, to see that $\calM^c$ is EF for chores, for all $i, j \in N$,
\begin{align*}
    c_i(\calM_i^c(\bfc))
    &= c_i(\calM_i^g(\bfv))\\
    &\overset{(a)}{=} (p + q) \cdot |\calM_i^g(\bfv)| - v_i(\calM_i^g(\bfv))\\
    &\overset{(b)}{\geq} (p + q) \cdot |\calM_j^g(\bfv)| - v_i(\calM_j^g(\bfv)) \\
    &\overset{(c)}{=} c_i(\calM_j^g(\bfv))
    = c_i(\calM_j^c(\bfc)),
\end{align*}
where $(a)$ and $(c)$ holds by \eqref{eqn:balance-value}, and $(b)$ holds due to the EF of $\calM^g$ for goods and the fact that $|\calM_i^g(\bfv)| = |\calM_j^g(\bfv)| = L$.

Finally, we show that $\calM^c$ is PO for chores in the following lemma, the proof of which is largely inspired by the proof of \cite[Proposition 6.3]{DBLP:journals/corr/abs-2407-13634}.

\begin{lemma}\label{lem:bobw-po}
    $\calM^c$ is PO for chores.
\end{lemma}

\begin{proof}
    We start by introducing the notion of \emph{market equilibrium}, which is a prevalent tool for proving PO, from the perspective of chores.
    A \emph{Fisher market} comprises of $n$ agents and $m$ divisible items, where each agent $i$ has a cost function $c_i(\cdot)$ and a budget $b_i \geq 0$.
    An outcome of the market is denoted by a pair $(\bfx, \bfr)$, where $\bfx$ is a fractional allocation and $\bfr \in \R_{\geq 0}^m$ is the price vector with $r(o)$ indicating the price of item $o$.
    We say that an outcome $(\bfx, \bfr)$ forms a \emph{market equilibrium} if
    \begin{itemize}
        \item all items are completely allocated: $\sum_{i \in N} x_{i}(o) = 1$ for every $o \in O$,

        \item each agent spends all his budget: $b_i = \sum_{o \in O} r(o) \cdot x_{i}(o)$ for every $i \in N$, and

        \item each agent $i$'s bundle only contains items $o$ with the \emph{bang-per-buck} ratio $c_i(o) / r(o)$ equal to his minimum bang-per-buck ratio $\gamma_i := \min_{o' \in O} c_i(o') / r(o')$: $x_{i}(o) > 0$ implies $c_i(o) / r(o) = \gamma_i$.
    \end{itemize}
    It is known that by the \emph{first welfare theorem}, if an outcome $(\bfx, \bfr)$ is a market equilibrium, then $\bfx$ satisfies PO for chores \cite{bogomolnaia2017competitive}.

    Now, we construct a price vector $\bfr$ for the allocation $\bfx = \calM^c(\bfc)$ and assign a budget to each agent so that $(\bfx, \bfr)$ constitutes a market equilibrium.
    Recall that $\bfv$ is the bi-valued utility profile with $v_i(o) = (p + q) - c_i(o)$ for all $i \in N$ and $o \in O$, $\bfx'$ is the allocation that maximizes $f_{\bfv}(\cdot)$ defined by \eqref{eqn:nash-welfare} with items in $H_{\bfv}$ unallocated, and $\bfx''$ is the allocation obtained from $\bfx'$ by truncating the bundles $x_i'$ with $|x_i'| > L$.
    Let $Z = \{i \in N \mid |x_i'| > L \}$ denote the set of agents whose bundles are truncated.
    For each item $o \in O \setminus H_{\bfv}$, we know by the first statement in \Cref{lem:char-mnw} that in $\bfx'$, either $o$ is fully allocated to agents in $N \setminus Z$, in which case we set the price $r(o) = q$, or $o$ is fully allocated to agents in $Z$, in which case we set $r(o) = p$.
    Moreover, for each item $o \in H_{\bfv}$, we set $r(o) = p$.
    Finally, each agent $i$'s budget $b_i$ is defined as the total price of $x_i$, i.e.,
    \begin{align*}
        b_i = 
        \begin{cases}
            q \cdot |x_i'| + p \cdot (L - |x_i'|), & i \in N \setminus Z,\\
            p \cdot L, & i \in Z.
        \end{cases}
    \end{align*}

    To conclude that $(\bfx, \bfr)$ forms a market equilibrium, it suffices to show that each bundle $x_i$ only includes the items with the minimum bang-per-buck ratio $\gamma_i$ of agent $i$.
    Notice that each bundle $x_i'$ only contains the items $o$ satisfying $v_i(o) = p$, which implies that $c_i(o) = q$.
    Hence, for each agent $i \in Z$, the bang-per-buck of each item in $x_i$ is $q / p$, which is the minimum bang-per-buck ratio one can hope for.
    On the other hand, for each agent $i \in N \setminus Z$, we claim that $\gamma_i = 1$.
    To see this, by the construction of the price vector, it is sufficient to prove that each item $o$ in $x_j'$ for some $j \in Z$ satisfies $c_i(o) = p$, or equivalently, $v_i(o) = q$, which holds by the second statement in \Cref{lem:char-mnw}.
    Finally, it is easy to verify that agent $i$ only receives items with a bang-per-buck ratio of $1$ for him, concluding the proof.
\end{proof}

\subsection{Proof of Lemma~\ref{lem:decomp-ps-plus}}
\label{sec:proof-decom-ps}

We first recall the following result by \cite{DBLP:conf/wine/Aziz20} to implement fractional allocations outputted by \PSg.

\begin{lemma}[\cite{DBLP:conf/wine/Aziz20}]\label{lem:decomp-ps}
    Suppose that agents have additive utility functions.
    For every fractional allocation $\bfx$ corresponding to an outcome of \PSg, we can find a random integral allocation $\widehat{A}$ such that $\Pr[o \in \widehat{A}_i] = x_{i}(o)$ for all $i \in N$ and $o \in O$.
    Moreover, for each integral allocation $A$ in the support of $\widehat{A}$, it satisfies that $|A_i| \leq |A_j| + 1$ for all $i, j \in N$, and we can label the items in each bundle $A_i$ as $o_i^1, o_i^2, \ldots, o_i^{|A_i|}$ such that
    \begin{enumerate}
        \item $v_i(o_i^k) \geq v_i(o_i^{k + 1})$ for all $i \in N$ and $k \in [|A_i| - 1]$, and
        \item $v_i(o_i^k) \geq v_i(o_j^{k + 1})$ for all $i, j \in N$ and $k \in [\min\{|A_i|, |A_j| - 1\}]$.
    \end{enumerate}
\end{lemma}

Let $\bfc$ be an additive cost profile and $\bfx$ be a fractional allocation corresponding to an outcome of \PSc~with respect to $\bfc$.
Let $c_{\max} = \max_{i \in N, o \in O} c_i(o) > 0$, and we define an additive utility profile $\bfv$ with $v_i(o) = c_{\max} - c_i(o)$ for all $i \in N$ and $o \in O$.
Since $v_i(o) \geq v_i(o')$ iff $c_i(o) \leq c_i(o')$ for all $i \in N$ and $o, o' \in O$, it suffices to find a random integral allocation $\widehat{A}$ such that $\Pr[o \in \widehat{A}_i] = x_{i}(o)$ for all $i \in N$ and $o \in O$.
Moreover, we will show that the given $\widehat{A}$ also satisfies that for each integral allocation $A$ in the support of $\widehat{A}$, $|A_i| \leq |A_j| + 1$ for all $i, j \in N$, and we can label the items in each bundle $A_i$ as $o_i^1, o_i^2, \ldots, o_i^{|A_i|}$ such that
\begin{enumerate}
    \item \label{item:con1-good} $v_i(o_i^k) \geq v_i(o_i^{k + 1})$ for all $i \in N$ and $k \in [|A_i| - 1]$,
    \item \label{item:con2-good} for all $i, j \in N$ with $|A_i| \leq |A_j|$, for every $k \in [|A_j| - 1]$, $v_i(o_i^k) \geq v_i(o_j^{k + 1})$, and
    \item \label{item:con3-good} for all $i, j \in N$ with $|A_i| > |A_j|$, for every $k \in [|A_j|]$, $v_i(o_i^k) \geq v_i(o_j^k)$.
\end{enumerate}

    Notice that $\bfx$ is an outcome of \PSg~with respect to $\bfv$ due to the construction of $\bfv$ and the fact that $\bfx$ is an outcome of \PSc~with respect to $\bfc$.
    If $m$ is a multiple of $n$, then the random allocation given by \Cref{lem:decomp-ps} satisfies all desired conditions since all bundles in the resulting integral allocation are of the same size.
    Now, we focus on the case where $m = kn + r$ for some $k \in \mathbb{Z}_{\geq 0}$ and $r \in \{1, \ldots, n - 1\}$.
    We create a set of $n - r$ new items, denoted as $O' = \{o_1', \ldots, o_{n - r}'\}$, and set $v_i(o) = 2c_{\max} > c_{\max}$ for all $i \in N$ and $o \in O'$.
    Moreover, let $\bfx'$ be a fractional allocation over $O \cup O'$ defined as $x_i'(o) = x_i(o)$ for all $i \in N$ and $o \in O$, and $x_i'(o) = (n - r) / n$ for all $i \in N$ and $o \in O'$.
    Since $\bfx$ is an outcome of \PSg~and $v_i(o) \leq c_{\max} < v_i(o')$ for all $i \in N$, $o \in O$, and $o' \in O'$, it follows that $\bfx'$ is also an outcome of \PSg.
    Hence, we can apply \Cref{lem:decomp-ps} on $\bfx'$ to obtain a random integral allocation $\widehat{A}'$.
    Let $\widehat{A}$ be the random integral allocation obtained by restricting $\widehat{A}'$ on $O$, i.e., $\widehat{A}_i = \widehat{A}_i' \cap O$ for every $i \in N$, and we show that $\widehat{A}$ meets our requirements.

    Let $A'$ be an arbitrary allocation in the support of $\widehat{A}'$, and we start by establishing some properties of $A'$.
    We label the items in each bundle $A_i'$ by $o_i^1, o_i^2, \ldots, o_i^{|A_i'|}$ as specified by \Cref{lem:decomp-ps}.
    Since $|O \cup O'| = (k + 1)n$, it follows that $|A_i'| = k + 1$ for every $i \in N$.
    Furthermore, we claim that $|A_i' \cap O'| \leq 1$ for every $i \in N$.
    To see this, assume for contradiction that $|A_i' \cap O'| \geq 2$ for some $i \in N$, which implies that $o_i^1, o_i^2 \in O'$.
    Let $j \in N$ satisfy $A_j' \cap O' = \emptyset$, which must exist since $|O'| < n$, and it holds that $o_j^1 \in O$.
    However, by the guarantee of \Cref{lem:decomp-ps}, $c_{\max} \geq v_j(o_j^1) \geq v_j(o_i^2) = 2c_{\max}$, leading to a contradiction.

    Now, we show that $\widehat{A}$ is a desired random allocation.
    Firstly, for all $i \in N$ and $o \in O$, $\Pr[o \in \widehat{A}_i] = \Pr[o \in \widehat{A}_i'] = x_{i}'(o) = x_{i}(o)$.
    Next, let $A'$ be an arbitrary allocation in the support of $\widehat{A}'$, and let $A$ be the allocation obtained by restricting $A'$ on $O$, i.e., $A_i = A_i' \cap O$ for every $i \in N$.
    For all $i, j \in N$, since $|A_j' \cap O'| \leq 1$, it holds that $|A_j| \geq |A_j'| - 1 = |A_i'| - 1 \geq |A_i| - 1$.
    Notice that $A$ satisfies Properties~\ref{item:con1-good} and~\ref{item:con2-good} due to the guarantee of $A'$ promised by \Cref{lem:decomp-ps}, and it remains to show that $A$ satisfies Property~\ref{item:con3}.
    Fix $i, j \in N$ with $|A_i| > |A_j|$, and it must hold that $A_i = \{o_i^1, o_i^2, \ldots, o_i^{c + 1}\}$ and $A_j = \{o_j^2, o_j^3, \ldots, o_j^{c + 1}\}$ with $o_j^1 \in O'$.
    Then, Property~\ref{item:con3} follows by the guarantee of $A'$ promised by \Cref{lem:decomp-ps} and the relabeling of items in $A_j$.

%% file: appendix/indiv-char.tex
\section{Characterizing Truthful Mechanisms for Indivisible Chores with Two Agents}
\label{sec:indiv-char}

In this section, we assume that there are $n = 2$ agents, and items are indivisible.
When agents' strategic behaviors are taken into consideration, \cite{DBLP:conf/sigecom/AmanatidisBCM17} characterize all truthful mechanisms for goods with two agents.
By leveraging \Cref{thm:trans-indiv-two-agents}, we obtain the counterpart of the characterization for chores.
In particular, we show that for two agents, the family of truthful mechanisms for chores admits a form symmetric to that for goods.

\subsection{Picking-Exchange Mechanisms}

We first define the related types of mechanisms for goods, which are proposed by \cite{DBLP:conf/sigecom/AmanatidisBCM17}, and their adaptations for chores.

\paragraph{Picking mechanisms.}
We define the family of \emph{picking mechanisms} where each agent makes a selection from the offers that the mechanism proposes to him.
Given a set of items $S$, \emph{a set of offers $\calP$ on $S$} is defined as a nonempty collection of subsets of $S$ that exactly covers $S$ (i.e., $\bigcup_{T \in \calP} T = S$) such that no element in $S$ appears in all subsets (i.e., $\bigcap_{T \in \calP} T = \emptyset$).

\begin{definition}
    A mechanism $\calM$ is a \emph{picking mechanism for goods (resp. chores)} if there exists a partition $(X_1, X_2)$ of $O$ and sets of offers $\calP_1$ and $\calP_2$ respectively on $X_1$ and $X_2$ such that for every utility profile $\bfv$ (resp. cost profile $\bfc$),
    \begin{align*}
        \calM_i(\bfv) \cap X_i \in \argmax_{S \in \calP_i} v_i(S) \quad \left( \text{resp. } \calM_i(\bfc) \cap X_i \in \argmin_{S \in \calP_i} c_i(S) \right)
    \end{align*}
    for every $i \in N$.
\end{definition}

\paragraph{Exchange mechanisms.}
We now define another family of mechanisms called \emph{exchange mechanisms}.
For two disjoint subsets $S, T$ of $O$, we refer to the ordered pair $(S, T)$ as an \emph{exchange deal}.
Given a utility profile $\bfv$ (resp. cost profile $\bfc$), we say that an exchange deal $(S, T)$ is \emph{favorable with respect to $\bfv$ for goods (resp. $\bfc$ for chores)} if $v_1(T) > v_1(S)$ and $v_2(T) < v_2(S)$ (resp. $c_1(T) < c_1(S)$ and $c_2(T) > c_2(S)$), and is \emph{unfavorable with respect to $\bfv$ for goods (resp. $\bfc$ for chores)} if $v_1(T) < v_1(S)$ or $v_2(T) > v_2(S)$ (resp. $c_1(T) > c_1(S)$ or $c_2(T) < c_2(S)$).
Let $S$ and $T$ be two disjoint subsets of $O$, and let $S_1, \ldots, S_k$ and $T_1, \ldots, T_k$ be two collections of nonempty and pairwise disjoint subsets of $S$ and $T$, respectively.
Then, we say that the set of exchange deals $D = \{(S_1, T_1), (S_2, T_2), \ldots, (S_k, T_k)\}$ on $(S, T)$ is \emph{valid}.

\begin{definition}
    A mechanism $\calM$ is an \emph{exchange mechanism for goods (resp. chores)} if there exists a partition $(Y_1, Y_2)$ of $O$ and a valid set of exchange deals $D = \{(S_1, T_1), \ldots, (S_k, T_k)\}$ on $(Y_1, Y_2)$ such that for every utility profile $\bfv$ (resp. cost profile $\bfc$), there exists a set of indices $I = I(\bfv) \subseteq [k]$ (resp. $I \subseteq I(\bfc) \subseteq [k]$) satisfying
    \begin{align*}
        \calM_1(\bfv) = \left( Y_1 \setminus \bigcup_{i \in I} S_i \right) \cup \bigcup_{i \in I} T_i
        \quad \left( \text{resp. } \calM_1(\bfc) = \left( Y_1 \setminus \bigcup_{i \in I} S_i \right) \cup \bigcup_{i \in I} T_i  \right).
    \end{align*}
    Moreover, $I$ contains the indices of every favorable exchange deal with respect to $\bfv$ for goods (resp. $\bfc$ for chores), but no indices of unfavorable exchange deals with respect to $\bfv$ for goods (resp. $\bfc$ for chores).
\end{definition}

\paragraph{Picking-exchange mechanisms.}
Now, we move to the family of picking-exchange mechanisms, which generalizes both picking and exchange mechanisms.

\begin{definition}
    A mechanism $\calM$ is a \emph{picking-exchange mechanism for goods (resp. chores)} if there exists a partition $(X_1, X_2, Y_1, Y_2)$ of $O$, sets of offers $\calP_1$ and $\calP_2$ respectively on $X_1$ and $X_2$, and a valid set of exchange deals $D = \{(S_1, T_1), \ldots, (S_k, T_k)\}$ on $(Y_1, Y_2)$, such that for all utility profile $\bfv$ (resp. cost profile $\bfc$) and $i \in N$,
    \begin{align*}
        \calM_i(\bfv) \cap X_i \in \argmax_{S \in \calP_i} v_i(S) \quad \left( \text{resp. } \calM_i(\bfc) \cap X_i \in \argmin_{S \in \calP_i} c_i(S) \right),
    \end{align*}
    and
    \begin{align*}
        \calM_1(\bfv) \cap (Y_1 \cup Y_2) = \left( Y_1 \setminus \bigcup_{i \in I} S_i \right) \cup \bigcup_{i \in I} T_i
        \quad \left( \text{resp. } \calM_1(\bfc) \cap (Y_1 \cup Y_2) = \left( Y_1 \setminus \bigcup_{i \in I} S_i \right) \cup \bigcup_{i \in I} T_i \right),
    \end{align*}
    where $I = I(\bfv) \subseteq [k]$ (resp. $I = I(\bfc) \subseteq [k]$) contains the indices of all favorable exchange deals with respect to $\bfv$ for goods (resp. $\bfc$ for chores), but no indices of unfavorable exchange deals with respect to $\bfv$ for goods (resp. $\bfc$ for chores).
\end{definition}

It is easy to verify that every picking-exchange mechanism for goods (resp. chores) is truthful for goods (resp. chores).
Moreover, as discussed by \cite{DBLP:conf/sigecom/AmanatidisBCM17}, it is instructive to treat a picking-exchange mechanism as independently running a picking mechanism on $(X_1, X_2)$ and an exchange mechanism on $(Y_1, Y_2)$.
However, this is not necessarily true in general, as the tie-breaking of choosing offers in the picking mechanism might be correlated with the decision of whether to perform an exchange deal that is neither favorable nor unfavorable.
We refer to \cite{DBLP:conf/sigecom/AmanatidisBCM17} for more comprehensive discussions and illustrations (see, e.g., \cite[Example 3]{DBLP:conf/sigecom/AmanatidisBCM17}) on the families of mechanisms defined above.

\subsection{Characterization}

Recall that we have the following characterization of truthful mechanisms for goods.

\begin{theorem}[\cite{DBLP:conf/sigecom/AmanatidisBCM17}]\label{thm:truth-mecha-goods}
    For indivisible items with two agents, a mechanism is truthful for goods iff it is a picking-exchange mechanism for goods.
\end{theorem}

We provide an analog of \Cref{thm:truth-mecha-goods} for chores by combining it with \Cref{thm:trans-indiv-two-agents}.

\begin{theorem}[Restatement of \Cref{thm:char-truth-chores-informal}]\label{thm:truth-mecha-chores}
    For indivisible items with two agents, a mechanism is truthful for chores iff it is a picking-exchange mechanism for chores.
\end{theorem}

\begin{proof}
    It is easy to verify that every picking-exchange mechanism for chores is truthful for chores.
    Hence, it remains to show that every truthful mechanism for chores can be implemented as a picking-exchange mechanism for chores.
    Fix a truthful mechanism $\calM^c$ for chores.
    Let $\calM^g$ be the mechanism satisfying $\calM^g(\bfc) = (\calM^c_2(\bfc), \calM^c_1(\bfc))$ for every cost profile $\bfc$, and $\calM^g$ is truthful for goods by \Cref{thm:trans-indiv-two-agents} and the truthfulness of $\calM^c$ for chores.
    Furthermore, by \Cref{thm:truth-mecha-goods}, $\calM^g$ can be implemented as a picking-exchange mechanism for goods.
    As a result, there exists a partition $(X_1, X_2, Y_1, Y_2)$ of $O$, sets of offers $\calP_1$ and $\calP_2$ respectively on $X_1$ and $X_2$, and a valid set of exchange deals $D = \{(S_1, T_1), \ldots, (S_k, T_k)\}$ on $(Y_1, Y_2)$, such that for all cost profile $\bfc$ and $i \in N$,
    \begin{align*}
        \calM^g_i(\bfc) \cap X_i \in \argmax_{S \in \calP_i} c_i(S) \quad \text{and} \quad
        \calM^g_1(\bfc) \cap (Y_1 \cup Y_2) = \left( Y_1 \setminus \bigcup_{i \in I} S_i \right) \cup \bigcup_{i \in I} T_i,
    \end{align*}
    where $I = I(\bfc) \subseteq [k]$ contains the indices of all favorable exchange deals with respect to $\bfc$ for goods, but no indices of unfavorable exchange deals with respect to $\bfc$ for goods.

    Now, we are ready to show that $\calM^c$ can be implemented as a picking-exchange mechanism for chores.
    Define new sets of offers by $\calP_1' = \{X_1 \setminus S \mid S \in \calP_1\}$ and $\calP_2' = \{X_2 \setminus S \mid S \in \calP_2\}$ respectively on $X_1$ and $X_2$, and define a valid set of exchange deals $D' = \{(T_1, S_1), \ldots, (T_k, S_k)\}$ on $(Y_2, Y_1)$.
    For all cost profile $\bfc$ and $i \in N$, it holds that
    \begin{align*}
        \calM^c_i(\bfc) \cap X_i
        &= (O \setminus \calM_i^g(\bfc)) \cap X_i
        = X_i \setminus (\calM_i^g(\bfc) \cap X_i)\\
        &\in \argmin_{S \in \calP_i} c_i(X_i \setminus S)
        = \argmin_{S' \in \calP_i'} c_i(S'),
    \end{align*}
    and
    \begin{align*}
        \calM_1^c(\bfc) \cap (Y_1 \cup Y_2)
        &= (O \setminus \calM_1^g(\bfc)) \cap (Y_1 \cup Y_2)
        = (Y_1 \cup Y_2) \setminus (\calM_1^g(\bfc) \cap (Y_1 \cup Y_2))\\
        &= (Y_1 \cup Y_2) \setminus \left( \left( Y_1 \setminus \bigcup_{i \in I} S_i \right) \cup \bigcup_{i \in I} T_i \right)
        = \left( Y_2 \setminus \bigcup_{i \in I} T_i \right) \cup \bigcup_{i \in I} S_i,
    \end{align*}
    where $I = I(\bfc)$.
    Since an exchange deal $(S, T)$ is favorable (resp. unfavorable) for goods iff the exchange deal $(T, S)$ is favorable (resp. unfavorable) for chores, it follows that $I$ contains the indices of all favorable exchange deals in $D'$ with respect to $\bfc$ for chores, but no indices of unfavorable exchange deals in $D'$ with respect to $\bfc$ for chores.
    Combining all the above analysis concludes that $\calM^c$ can be implemented as a picking-exchange mechanism for chores.
\end{proof}

%% file: appendix/proof-div.tex
\section{Proof of \Cref{lem:anony-item-symme}}
\label{sec:proof-div}

    Let $\calM$ be a truthful mechanism with an efficiency ratio of $\delta$ for chores, and let $\calM'$ be the mechanism satisfying
    \begin{align*}
        \calM_{1, o}'(c_1, c_2) = \frac{1}{2m!} \sum_{\sigma \in \Sym(O)} \left( \calM_{1, \sigma(o)} (c_1^{\sigma}, c_2^{\sigma}) + \calM_{2, \sigma(o)}(c_2^{\sigma}, c_1^{\sigma}) \right)
    \end{align*}
    and
    \begin{align*}
        \calM_{2, o}'(c_1, c_2) = \frac{1}{2m!} \sum_{\sigma \in \Sym(O)} \left( \calM_{2, \sigma(o)} (c_1^{\sigma}, c_2^{\sigma}) + \calM_{1, \sigma(o)}(c_2^{\sigma}, c_1^{\sigma}) \right)
    \end{align*}
    for all cost profile $\bfc = (c_1, c_2)$ and $o \in O$.
    It is easy to verify by definition that $\calM'$ is anonymous and item-symmetric.
    To see that $\calM'$ is truthful for chores, for all cost profile $\bfc = (c_1, c_2)$ and cost function $\hat{c}_1$,
    \begin{align*}
        c_1(\calM_1'(\hat{c}_1, c_2))
        &= \frac{1}{2m!} \sum_{\sigma \in \Sym(O)} \left( c_1^{\sigma}(\calM_1(\hat{c}_1^{\sigma}, c_2^{\sigma})) + c_1^{\sigma}(\calM_2(c_2^{\sigma}, \hat{c}_1^{\sigma})) \right)\\
        &\geq \frac{1}{2m!} \sum_{\sigma \in \Sym(O)} \left( c_1^{\sigma}(\calM_1(c_1^{\sigma}, c_2^{\sigma})) + c_1^{\sigma}(\calM_2(c_2^{\sigma}, c_1^{\sigma})) \right) \\
        &= c_1(\calM_1'(c_1, c_2)),
    \end{align*}
    where the inequality holds by the truthfulness of $\calM$ for chores.
    We can analogously show that agent $2$ cannot decrease his cost via misreporting under $\calM'$, concluding that $\calM'$ is truthful for chores.
    Finally, the efficiency ratio of $\calM'$ is given by
    \begin{align*}
        \inf_{\bfc} \frac{\SW^c(\bfc)}{\SW(\calM'(\bfc))}  
        &= \inf_{\bfc} \frac{\SW^c(\bfc)}{\frac{1}{2m!}\sum_{\sigma \in \Sym(O)} \left( \SW(\calM(c_1^{\sigma}, c_2^{\sigma})) + \SW(\calM(c_2^{\sigma}, c_1^{\sigma})) \right)}\\
        &\geq \inf_{\bfc} \frac{\SW^c(\bfc)}{\max \left \{ \max_{\sigma \in \Sym(O)} \SW(\calM(c_1^{\sigma}, c_2^{\sigma})), \max_{\sigma \in \Sym(O)} \SW(\calM(c_2^{\sigma}, c_1^{\sigma}))\right \}}\\
        &= \inf_{\bfc} \frac{\SW^c(\bfc)}{\SW(\calM(\bfc))}\\
        &= \delta,
    \end{align*}
    which concludes the proof.

%% file: main.bbl
\newcommand{\etalchar}[1]{$^{#1}$}
\begin{thebibliography}{ABF{\etalchar{+}}21b}

\bibitem[AAB{\etalchar{+}}23]{Amanatidis_2023}
Georgios Amanatidis, Haris Aziz, Georgios Birmpas, Aris Filos-Ratsikas, Bo~Li, Hervé Moulin, Alexandros~A. Voudouris, and Xiaowei Wu.
\newblock Fair division of indivisible goods: Recent progress and open questions.
\newblock {\em Artif. Intell.}, 322:103965, 2023.

\bibitem[ABCM17]{DBLP:conf/sigecom/AmanatidisBCM17}
Georgios Amanatidis, Georgios Birmpas, George Christodoulou, and Evangelos Markakis.
\newblock Truthful allocation mechanisms without payments: Characterization and implications on fairness.
\newblock In {\em {EC}}, pages 545--562. {ACM}, 2017.

\bibitem[ABF{\etalchar{+}}21a]{DBLP:journals/tcs/AmanatidisBFHV21}
Georgios Amanatidis, Georgios Birmpas, Aris Filos{-}Ratsikas, Alexandros Hollender, and Alexandros~A. Voudouris.
\newblock Maximum nash welfare and other stories about {EFX}.
\newblock {\em Theor. Comput. Sci.}, 863:69--85, 2021.

\bibitem[ABF{\etalchar{+}}21b]{DBLP:wine/AmanatidisBFLLR21}
Georgios Amanatidis, Georgios Birmpas, Federico Fusco, Philip Lazos, Stefano Leonardi, and Rebecca Reiffenh{\"{a}}user.
\newblock Allocating indivisible goods to strategic agents: Pure nash equilibria and fairness.
\newblock In {\em {WINE}}, volume 13112 of {\em Lecture Notes in Computer Science}, pages 149--166. Springer, 2021.

\bibitem[ABL{\etalchar{+}}23]{DBLP:conf/sigecom/AmanatidisBL0R23}
Georgios Amanatidis, Georgios Birmpas, Philip Lazos, Stefano Leonardi, and Rebecca Reiffenh{\"{a}}user.
\newblock Round-robin beyond additive agents: Existence and fairness of approximate equilibria.
\newblock In {\em {EC}}, pages 67--87. {ACM}, 2023.

\bibitem[ABM16]{DBLP:conf/ijcai/AmanatidisBM16}
Georgios Amanatidis, Georgios Birmpas, and Evangelos Markakis.
\newblock On truthful mechanisms for maximin share allocations.
\newblock In {\em {IJCAI}}, pages 31--37. {IJCAI/AAAI} Press, 2016.

\bibitem[AFSV23]{aziz2023best}
Haris Aziz, Rupert Freeman, Nisarg Shah, and Rohit Vaish.
\newblock Best of both worlds: Ex ante and ex post fairness in resource allocation.
\newblock {\em Operations Research}, 2023.

\bibitem[ALW19]{DBLP:conf/ijcai/00010W19}
Haris Aziz, Bo~Li, and Xiaowei Wu.
\newblock Strategyproof and approximately maxmin fair share allocation of chores.
\newblock In {\em {IJCAI}}, pages 60--66. ijcai.org, 2019.

\bibitem[ALW24]{DBLP:journals/mp/AzizLW24}
Haris Aziz, Bo~Li, and Xiaowei Wu.
\newblock Approximate and strategyproof maximin share allocation of chores with ordinal preferences.
\newblock {\em Math. Program.}, 203(1):319--345, 2024.

\bibitem[AY14]{DBLP:conf/wine/AzizY14}
Haris Aziz and Chun Ye.
\newblock Cake cutting algorithms for piecewise constant and piecewise uniform valuations.
\newblock In {\em {WINE}}, volume 8877 of {\em Lecture Notes in Computer Science}, pages 1--14. Springer, 2014.

\bibitem[Azi20]{DBLP:conf/wine/Aziz20}
Haris Aziz.
\newblock Simultaneously achieving ex-ante and ex-post fairness.
\newblock In {\em {WINE}}, volume 12495 of {\em Lecture Notes in Computer Science}, pages 341--355. Springer, 2020.

\bibitem[BD05]{DBLP:journals/sigecom/BezakovaD05}
Ivona Bez{\'{a}}kov{\'{a}} and Varsha Dani.
\newblock Allocating indivisible goods.
\newblock {\em SIGecom Exch.}, 5(3):11--18, 2005.

\bibitem[BEF21]{DBLP:conf/aaai/BabaioffEF21}
Moshe Babaioff, Tomer Ezra, and Uriel Feige.
\newblock Fair and truthful mechanisms for dichotomous valuations.
\newblock In {\em {AAAI}}, pages 5119--5126. {AAAI} Press, 2021.

\bibitem[BESS25]{DBLP:journals/siamdm/BeiESS25}
Xiaohui Bei, Edith Elkind, Erel Segal{-}Halevi, and Warut Suksompong.
\newblock Dividing a graphical cake.
\newblock {\em {SIAM} J. Discret. Math.}, 39(1):19--54, 2025.

\bibitem[BHS20]{DBLP:journals/scw/BeiHS20}
Xiaohui Bei, Guangda Huzhang, and Warut Suksompong.
\newblock Truthful fair division without free disposal.
\newblock {\em Soc. Choice Welf.}, 55(3):523--545, 2020.

\bibitem[BJK{\etalchar{+}}06]{brams2006better}
Steven~J Brams, Michael~A Jones, Christian Klamler, et~al.
\newblock Better ways to cut a cake.
\newblock {\em Notices of the AMS}, 53(11):1314--1321, 2006.

\bibitem[BMSY17]{bogomolnaia2017competitive}
Anna Bogomolnaia, Herv{\'e} Moulin, Fedor Sandomirskiy, and Elena Yanovskaya.
\newblock Competitive division of a mixed manna.
\newblock {\em Econometrica}, 85(6):1847--1871, 2017.

\bibitem[BST23]{DBLP:journals/ai/BuST23}
Xiaolin Bu, Jiaxin Song, and Biaoshuai Tao.
\newblock On existence of truthful fair cake cutting mechanisms.
\newblock {\em Artif. Intell.}, 319:103904, 2023.

\bibitem[BT24]{DBLP:journals/corr/abs-2407-13634}
Xiaolin Bu and Biaoshuai Tao.
\newblock Truthful and almost envy-free mechanism of allocating indivisible goods: the power of randomness.
\newblock {\em CoRR}, abs/2407.13634, 2024.

\bibitem[BTWY25]{bei2025incentive}
Xiaohui Bei, Biaoshuai Tao, Jiajun Wu, and Mingwei Yang.
\newblock The incentive guarantees behind nash welfare in divisible resources allocation.
\newblock {\em Artificial Intelligence}, page 104335, 2025.

\bibitem[Bud11]{DBLP:conf/bqgt/Budish10}
Eric Budish.
\newblock The combinatorial assignment problem: approximate competitive equilibrium from equal incomes.
\newblock {\em J. Political Econ.}, 119(6):1061--1103, 2011.

\bibitem[BV22]{DBLP:conf/aaai/BarmanV22}
Siddharth Barman and Paritosh Verma.
\newblock Truthful and fair mechanisms for matroid-rank valuations.
\newblock In {\em {AAAI}}, pages 4801--4808. {AAAI} Press, 2022.

\bibitem[CGG13a]{DBLP:conf/sigecom/ColeGG13}
Richard Cole, Vasilis Gkatzelis, and Gagan Goel.
\newblock Mechanism design for fair division: allocating divisible items without payments.
\newblock In {\em {EC}}, pages 251--268. {ACM}, 2013.

\bibitem[CGG13b]{DBLP:conf/atal/ColeGG13}
Richard Cole, Vasilis Gkatzelis, and Gagan Goel.
\newblock Positive results for mechanism design without money.
\newblock In {\em {AAMAS}}, pages 1165--1166. {IFAAMAS}, 2013.

\bibitem[CGMM22]{DBLP:conf/sigecom/ChaudhuryGMM22}
Bhaskar~Ray Chaudhury, Jugal Garg, Peter McGlaughlin, and Ruta Mehta.
\newblock Competitive equilibrium with chores: Combinatorial algorithm and hardness.
\newblock In {\em {EC}}, pages 1106--1107. {ACM}, 2022.

\bibitem[Che16]{DBLP:conf/ijcai/Cheung16}
Yun~Kuen Cheung.
\newblock Better strategyproof mechanisms without payments or prior - an analytic approach.
\newblock In {\em {IJCAI}}, pages 194--200. {IJCAI/AAAI} Press, 2016.

\bibitem[CKM{\etalchar{+}}19]{DBLP:journals/teco/CaragiannisKMPS19}
Ioannis Caragiannis, David Kurokawa, Herv{\'{e}} Moulin, Ariel~D. Procaccia, Nisarg Shah, and Junxing Wang.
\newblock The unreasonable fairness of maximum nash welfare.
\newblock {\em {ACM} Trans. Economics and Comput.}, 7(3):12:1--12:32, 2019.

\bibitem[CLPP13]{DBLP:journals/geb/ChenLPP13}
Yiling Chen, John~K. Lai, David~C. Parkes, and Ariel~D. Procaccia.
\newblock Truth, justice, and cake cutting.
\newblock {\em Games Econ. Behav.}, 77(1):284--297, 2013.

\bibitem[EPS22]{DBLP:conf/atal/EbadianP022}
Soroush Ebadian, Dominik Peters, and Nisarg Shah.
\newblock How to fairly allocate easy and difficult chores.
\newblock In {\em {AAMAS}}, pages 372--380. International Foundation for Autonomous Agents and Multiagent Systems {(IFAAMAS)}, 2022.

\bibitem[Fol67]{foley1967resource}
Duncan Foley.
\newblock Resource allocation and the public sector.
\newblock {\em Yale Economic Essays}, pages 45--98, 1967.

\bibitem[FWVP23]{freeman2023equivalence}
Rupert Freeman, Jens Witkowski, Jennifer~Wortman Vaughan, and David~M Pennock.
\newblock An equivalence between fair division and wagering mechanisms.
\newblock {\em Management Science}, 2023.

\bibitem[GC10]{DBLP:conf/atal/GuoC10}
Mingyu Guo and Vincent Conitzer.
\newblock Strategy-proof allocation of multiple items between two agents without payments or priors.
\newblock In {\em {AAMAS}}, pages 881--888. {IFAAMAS}, 2010.

\bibitem[GIK{\etalchar{+}}24]{DBLP:journals/geb/GokoIKMSTYY24}
Hiromichi Goko, Ayumi Igarashi, Yasushi Kawase, Kazuhisa Makino, Hanna Sumita, Akihisa Tamura, Yu~Yokoi, and Makoto Yokoo.
\newblock A fair and truthful mechanism with limited subsidy.
\newblock {\em Games Econ. Behav.}, 144:49--70, 2024.

\bibitem[GM23]{DBLP:journals/tcs/GargM23}
Jugal Garg and Aniket Murhekar.
\newblock Computing fair and efficient allocations with few utility values.
\newblock {\em Theor. Comput. Sci.}, 962:113932, 2023.

\bibitem[GMQ22]{DBLP:conf/aaai/GargMQ22}
Jugal Garg, Aniket Murhekar, and John Qin.
\newblock Fair and efficient allocations of chores under bivalued preferences.
\newblock In {\em {AAAI}}, pages 5043--5050. {AAAI} Press, 2022.

\bibitem[GMQ23]{DBLP:conf/ijcai/GargMQ23}
Jugal Garg, Aniket Murhekar, and John Qin.
\newblock New algorithms for the fair and efficient allocation of indivisible chores.
\newblock In {\em {IJCAI}}, pages 2710--2718. ijcai.org, 2023.

\bibitem[GMT14]{DBLP:conf/ecai/GourvesMT14}
Laurent Gourv{\`{e}}s, J{\'{e}}r{\^{o}}me Monnot, and Lydia Tlilane.
\newblock Near fairness in matroids.
\newblock In {\em {ECAI}}, volume 263 of {\em Frontiers in Artificial Intelligence and Applications}, pages 393--398. {IOS} Press, 2014.

\bibitem[HPPS20]{DBLP:conf/wine/0002PP020}
Daniel Halpern, Ariel~D. Procaccia, Alexandros Psomas, and Nisarg Shah.
\newblock Fair division with binary valuations: One rule to rule them all.
\newblock In {\em {WINE}}, volume 12495 of {\em Lecture Notes in Computer Science}, pages 370--383. Springer, 2020.

\bibitem[HST25]{DBLP:journals/corr/abs-2502-18805}
Eden Hartman, Erel Segal{-}Halevi, and Biaoshuai Tao.
\newblock It's not all black and white: Degree of truthfulness for risk-avoiding agents.
\newblock {\em CoRR}, abs/2502.18805, 2025.

\bibitem[HSTZ11]{DBLP:conf/wine/HanSTZ11}
Li~Han, Chunzhi Su, Linpeng Tang, and Hongyang Zhang.
\newblock On strategy-proof allocation without payments or priors.
\newblock In {\em {WINE}}, volume 7090 of {\em Lecture Notes in Computer Science}, pages 182--193. Springer, 2011.

\bibitem[HWWZ24]{DBLP:journals/jcss/HuangWWZ24}
Haoqiang Huang, Zihe Wang, Zhide Wei, and Jie Zhang.
\newblock Bounded incentives in manipulating the probabilistic serial rule.
\newblock {\em J. Comput. Syst. Sci.}, 140:103491, 2024.

\bibitem[LMMS04]{DBLP:conf/sigecom/LiptonMMS04}
Richard~J. Lipton, Evangelos Markakis, Elchanan Mossel, and Amin Saberi.
\newblock On approximately fair allocations of indivisible goods.
\newblock In {\em {EC}}, pages 125--131. {ACM}, 2004.

\bibitem[LSX24]{DBLP:conf/atal/0037SX24}
Bo~Li, Ankang Sun, and Shiji Xing.
\newblock Bounding the incentive ratio of the probabilistic serial rule.
\newblock In {\em {AAMAS}}, pages 1128--1136. {{IFAAMAS}}, 2024.

\bibitem[MP11]{DBLP:conf/wine/MarkakisP11}
Evangelos Markakis and Christos{-}Alexandros Psomas.
\newblock On worst-case allocations in the presence of indivisible goods.
\newblock In {\em {WINE}}, volume 7090 of {\em Lecture Notes in Computer Science}, pages 278--289. Springer, 2011.

\bibitem[OSH22]{ortega2022obvious}
Josu{\'e} Ortega and Erel Segal-Halevi.
\newblock Obvious manipulations in cake-cutting.
\newblock {\em Social Choice and Welfare}, pages 1--20, 2022.

\bibitem[PV22]{DBLP:conf/nips/0001V22}
Alexandros Psomas and Paritosh Verma.
\newblock Fair and efficient allocations without obvious manipulations.
\newblock In {\em NeurIPS}, 2022.

\bibitem[SC25]{DBLP:journals/eor/SunC25}
Ankang Sun and Bo~Chen.
\newblock Randomized strategyproof mechanisms with best of both worlds fairness and efficiency.
\newblock {\em Eur. J. Oper. Res.}, 324(3):941--952, 2025.

\bibitem[Seg20]{DBLP:journals/aamas/Segal-Halevi20}
Erel Segal{-}Halevi.
\newblock Competitive equilibrium for almost all incomes: existence and fairness.
\newblock {\em Auton. Agents Multi Agent Syst.}, 34(1):26, 2020.

\bibitem[Ste48]{steihaus1948problem}
Hugo Steinhaus.
\newblock The problem of fair division.
\newblock {\em Econometrica}, 16:101--104, 1948.

\bibitem[TY24]{tao2024fairtruthfulmechanismsadditive}
Biaoshuai Tao and Mingwei Yang.
\newblock Fair and almost truthful mechanisms for additive valuations and beyond, 2024.

\end{thebibliography}
